%% file: cav19-extended.tex
\documentclass[acmsmall, screen, nonacm]{acmart}

\usepackage[boxed]{algorithm2e}
\usepackage{array}
\usepackage{bussproofs}
\usepackage{centernot}
\usepackage{longtable}
\usepackage{listings}
\usepackage{microtype}
\usepackage{multicol}
\usepackage{multirow}
\usepackage{pdflscape}
\usepackage{shuffle}
\usepackage{stmaryrd}
\usepackage{tikz}
\usepackage{tikz-qtree}
\usepackage{wrapfig}
\usepackage{float}
\usepackage{caption}

\usetikzlibrary{automata}
\usetikzlibrary{positioning}

\title{Reductions for Automated Hypersafety Verification}
\author{Azadeh Farzan}
\affiliation{\institution{University of Toronto}}
\author{Anthony Vandikas}
\affiliation{\institution{University of Toronto}}
\authorsaddresses{}

\def\@proofindent{\noindent}

\renewcommand{\state}{\mathcal{S}t}
\newcommand{\assert}{\mathcal{A}}
\newcommand{\stmt}{\Sigma}

\newcommand{\bool}{\mathbb{B}}

\newcommand{\pow}[1]{{\mathcal{P}(#1)}}
\newcommand{\lang}[2][]{\mathcal{L}_{#1}(#2)}
\newcommand{\sem}[1]{\llbracket#1\rrbracket}
\newcommand{\floor}[1]{{\lfloor#1\rfloor}}

\newcommand{\linear}[1]{\mathcal{L}in(#1)}
\newcommand{\partition}[1]{\mathcal{P}art(#1)}

\newcommand{\sleep}{\mathrm{sleep}}
\newcommand{\ignore}{\mathrm{ignore}}
\newcommand{\reduce}{\mathrm{reduce}}
\newcommand{\inactive}{\mathrm{inactive}}
\newcommand{\tool}{{\sc Weaver}~}

 \def\azadeh#1{}
 \def\anthony#1{}

\begin{document}

\maketitle

\begin{abstract}
We propose an automated verification technique for hypersafety properties, which express sets of valid interrelations between multiple finite runs of a program. The key observation is that constructing a proof for a small representative set of the runs of the product program (i.e. the product of the several copies of the program by itself), called a {\em reduction}, is sufficient to formally prove the hypersafety property about the program. We propose an algorithm based on a counterexample-guided refinement loop that simultaneously searches for a reduction and a proof of the correctness for the reduction. We demonstrate that our tool \tool is very effective in verifying a diverse array of hypersafety properties for a diverse class of input programs.
\end{abstract}

\input{./cav19-extended/intro.tex}
\input{./cav19-extended/examples.tex}
\input{./cav19-extended/prelim.tex}
\input{./cav19-extended/program.tex}
\input{./cav19-extended/lta.tex}
\input{./cav19-extended/sleep.tex}
\input{./cav19-extended/algorithm.tex}
\input{./cav19-extended/eval.tex}
\input{./cav19-extended/related.tex}

\bibliographystyle{splncs04}
\bibliography{cav19}
\appendix
\input{./cav19-extended/appendix.tex}

\end{document}

%% file: cav19-extended/intro.tex

\section{Introduction}\label{sec:intro}
A hypersafety property describes the set of valid interrelations between multiple finite runs of a program. A $k$-safety property \cite{ClarksonS08} is a program safety property whose violation is witnessed by at least $k$ finite runs of a program.  Determinism is an example of such a property: non-determinism can only be witnessed by two runs of the program on the same input which produce two different outputs. This makes determinism an instance of a 2-safety property.

The vast majority of existing program verification methodologies are geared towards verifying standard (1-)safety properties. This paper proposes an approach to automatically reduce verification of $k$-safety to verification of 1-safety, and hence a way to leverage existing safety verification techniques for hypersafety verification. The most straightforward way to do this is via \emph{self-composition} \cite{BartheDR11}, where verification is performed on $k$ memory-disjoint copies of the program, sequentially composed one after another. Unfortunately, the proofs in these cases are often very verbose, since the full functionality of each copy has to be captured by the proof. Moreover, when it comes to automated verification, the invariants required to verify such programs are often well beyond the capabilities of modern solvers \cite{TerauchiA05} even for very simple programs and properties.

The more practical approach, which is typically used in manual or automated proofs of such properties, is to compose $k$ memory-disjoint copies of the program \emph{in parallel} (instead of in sequence), and then verify some \emph{reduced} program obtained by removing redundant traces from the program formed in the previous step.
This parallel product program can have many such reductions.  For example, the program formed from sequential self-composition is one such reduction of the parallel product program. Therefore, care must be taken to choose a ``good'' reduction that {\em admits a simple proof}. Many existing approaches limit themselves to a narrow class of reductions, such as the one where each copy of the program executes in lockstep \cite{BartheCK11,EilersMH18,SousaD16}, or define a general class of reductions, but do not provide algorithms with guarantees of covering the entire class \cite{BartheCK13,SousaD16}.

We propose a solution that combines the search for a safety proof with the search for an appropriate reduction, in a counterexample-based refinement loop. Instead of settling on a single reduction in advance, we try to verify the entire (possibly infinite) set of reductions simultaneously and terminate as soon as some reduction is successfully verified. If the proof is not currently strong enough to cover at least one of the represented program reductions, then an appropriate set of counterexamples are generated that guarantee progress towards a proof.

Our solution is language-theoretic. We propose a way to represent sets of reductions using infinite tree automata. The standard safety proofs are also represented using the same automata, which have the desired closure properties. This allows us to check if a candidate proof is in fact a proof for one of the represented program reductions, with reasonable efficiency.

Our approach is not uniquely applicable to hypersafety properties of sequential programs. Our proposed set of reductions naturally work well for concurrent programs, and can be viewed in the spirit of reduction-based methods such as those proposed in \cite{ElmasQT09,PopeeaRW14}. This makes our approach particularly appealing when it comes to verification of hypersafety properties of concurrent programs, for example, proving that a concurrent program is deterministic.  The parallel composition for hypersafety verification mentioned above and the parallel composition of threads inside the multi-threaded program are treated in a uniform way by our proof construction and checking algorithms. In summary:
\begin{itemize}
\item We present a counterexample-guided refinement loop that simultaneously searches for a proof and a program reduction in Section \ref{sec:algorithm}. This refinement loop relies on an efficient algorithm for proof checking based on the antichain method of \cite{WulfDHR06}, and strong theoretical progress guarantees.
\item We propose an automata-based approach to representing a class of program reductions for k-safety verification. In Section \ref{sec:lta} we describe the precise class of automata we use and show how their use leads to an effective proof checking algorithm incorporated in our refinement loop.
\item We demonstrate the efficacy of our approach in proving hypersafety properties of sequential and concurrent benchmarks in Section \ref{sec:eval}.
\end{itemize}

%% file: cav19-extended/examples.tex

\section{Illustrative Example}\label{sec:example}

\let\oldnl\nl
\newcommand{\nonl}{\renewcommand{\nl}{\let\nl\oldnl}}

We use a simple program \textsc{Mult}, that computes the product of two non-negative integers, to illustrate the challenges of verifying hypersafety properties and the type of proof that our approach targets. Consider the multiplication program in Figure \ref{fig:exm1}(i), and assume we want to prove that it is distributive over addition.

\begin{figure}[t]
\begin{center}
\includegraphics[scale=0.22]{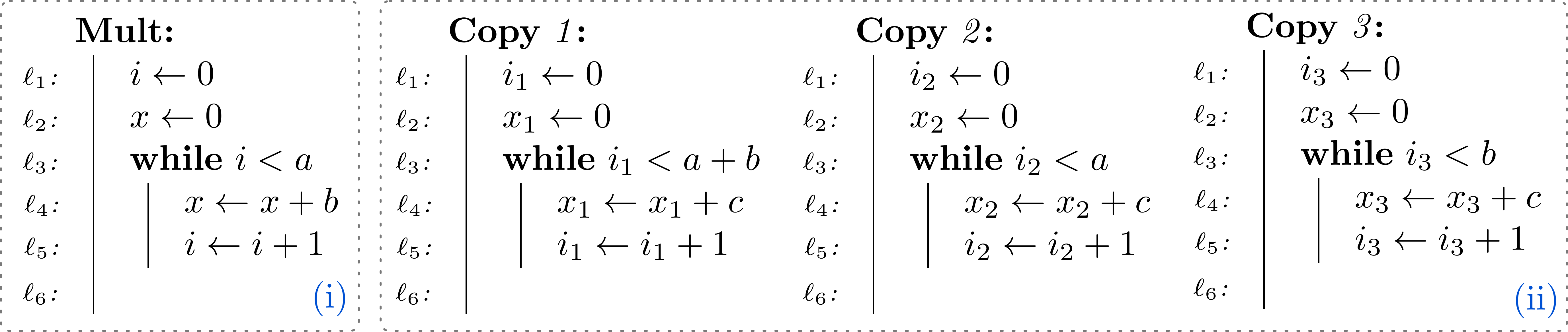}
\caption{Program \textsc{Mult} (i) and the parallel composition of three copies of it (ii).   \label{fig:exm1}} 
\end{center}
\end{figure}

In Figure \ref{fig:exm1} (ii), the parallel composition of \textsc{Mult} with two copies of itself is illustrated. The product program is formed for the purpose of proving distributivity, which can be encoded through the postcondition $x_1 = x_2 + x_3$. Since $a$, $b$, and $c$ are not modified in the program, the same variables are used across all copies. One way to prove \textsc{Mult} is distributive is to come up with an inductive invariant $\phi_{ijk}$ for each location in the product program, represented by a triple of program locations $(\ell_i, \ell_j, \ell_k)$,   such that $\mathit{true} \implies \phi_{111}$ and $\phi_{666} \implies x_1 = x_2 + x_3$. The main difficulty lies in finding assignments for locations such as $\phi_{611}$ that are points in the execution of the program where one thread has finished executing and the next one is starting. For example, at $(\ell_6, \ell_1, \ell_1)$ we need the assignment $\phi_{611} \gets x_1 = (a + b) * c$ which is non-linear. However, the program given in Figure \ref{fig:exm1}(ii) can be verified with simpler (linear) reasoning.


\begin{wrapfigure}{r}{0.28\textwidth}\vspace{-15pt}
  \begin{center}
  \includegraphics[scale=0.25]{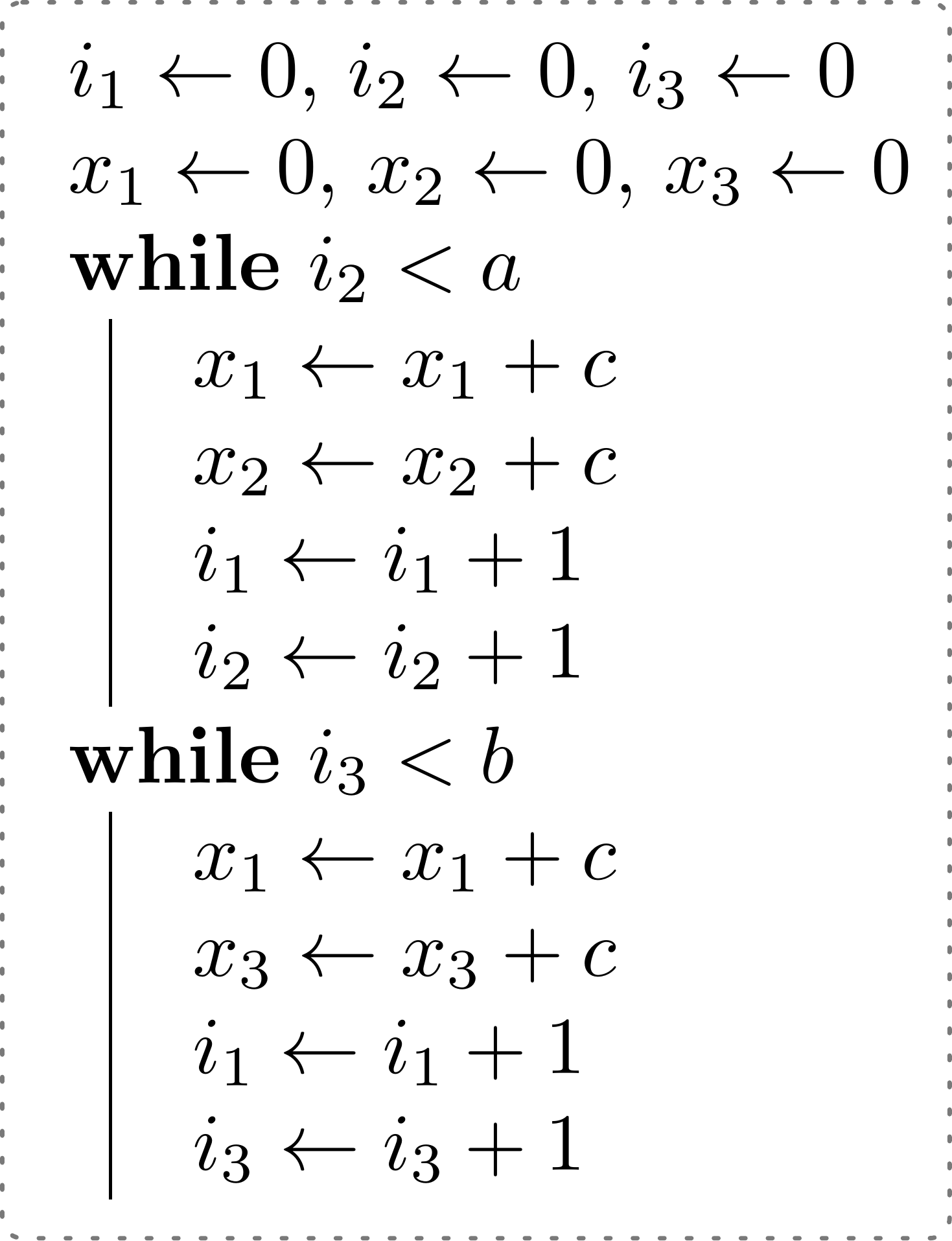}\vspace{-15pt}
  \end{center}
\end{wrapfigure}
The program on the right is a semantically equivalent \emph{reduction} of the full composition of Figure \ref{fig:exm1}(ii). Consider the program $P$ = (Copy 1 $||$ (Copy 2; Copy 3)). The program on the right is equivalent to a lockstep execution of the two parallel components of $P$. The validity of this reduction is derived from the fact that the statements in each thread are \emph{independent} of the statements in the other. That is, reordering the statements of different threads in an execution leads to an equivalent execution. It is easy to see that $x_1 = x_2 + x_3$ is an invariant of both while loops in the reduced program, and therefore, linear reasoning is sufficient to prove the postcondition for this program. Conceptually, this reduction (and its soundness proof) together with the proof of correctness for the reduced program constitute a proof that the original program \textsc{Mult} is distributive. Our proposed approach can come up with reductions like this and their corresponding proofs fully automatically. Note that a lockstep reduction of the program in Figure \ref{fig:exm1}(ii) would not yield a solution for this problem and therefore the discovery of the right reduction is an integral part of the solution.

%% file: cav19-extended/prelim.tex

\section{Programs and Proofs}

A non-deterministic finite automaton (NFA) is a tuple $A = (Q, \Sigma, \delta, q_0, F)$ where $Q$ is a finite set of states, $\Sigma$ is a finite alphabet, $\delta \subseteq Q \times \Sigma \times Q$ is the transition relation, $q_0 \in Q$ is the initial state, and $F \subseteq Q$ is the set of final states. A deterministic finite automaton (DFA) is an NFA whose transition relation is a function $\delta : Q \times \Sigma \to Q$. The language of an NFA or DFA $A$ is denoted $\lang{A}$, which is defined in the standard way \cite{Hopcroft:2006}.

\subsection{Program Traces}\label{sec:ptrace}

$\state$ denotes the (possibly infinite) set of \emph{program states}. For example, a program with two integer variables has $\state = \mathbb{Z} \times \mathbb{Z}$. $\assert \subseteq \state$ is a (possibly infinite) set of \emph{assertions} on program states.  $\stmt$ denotes a finite alphabet of program \emph{statements}. We refer to a finite string of statements as a (program) \emph{trace}. For each statement $a \in \stmt$ we associate a \emph{semantics} $\sem{a} \subseteq \state \times \state$ and extend $\sem{-}$ to traces via (relation) composition. A trace $x \in \stmt^*$ is said to be \emph{infeasible} if $\sem{x}(\state) = \emptyset$, where $\sem{x}(\state)$ denotes the image of $\sem{x}$ under $\state$.

To abstract away from a particular program syntax, we define a \emph{program} as a regular language of traces. The semantics of a program $P$ is simply the union of the semantics of its traces $\sem{P} = \bigcup_{x \in P} \sem{x}$. Concretely, one may obtain programs as languages by interpreting their edge-labelled control-flow graphs as DFAs: each vertex in the control flow graph is a state, and each edge in the control flow graph is a transition. The control flow graph entry location is the initial state of the DFA and all its exit locations are final states.

\subsection{Safety}\label{sec:safety}

There are many equivalent notions of program safety; we use non-reachability. A program $P$ is \emph{safe} if all traces of $P$ are infeasible, i.e. $\sem{P}(\state) = \emptyset$.
Standard partial correctness specifications are then represented via a simple encoding. Given a precondition $\phi$ and a postcondition $\psi$, the validity of the Hoare-triple $\{\phi\}P\{\psi\}$ is equivalent to the safety of $[\phi] \cdot P \cdot [\neg\psi]$, where $[]$ is a standard assume statement (or the singleton set containing it), and $\cdot$ is language concatenation.

\begin{example}
  We use determinism as an example of how $k$-safety can be encoded in the framework defined thus far. If $P$ is a program then determinism of $P$ is equivalent to safety of $[\phi] \cdot (P_1 \shuffle P_2) \cdot [\neg\phi]$  where $P_1$ and $P_2$ are copies of $P$ operating on disjoint variables, $\shuffle$ is a shuffle product of two languages, and $[\phi]$ is an assume statement asserting that the variables in each copy of $P$ are equal.
\end{example}

A \emph{proof} is a finite set of assertions $\Pi \subseteq \assert$ that includes $\mathit{true}$ and $\mathit{false}$. Each $\Pi$ gives rise to an NFA $\Pi_{NFA} = (\Pi, \state, \delta_\Pi, \mathit{true}, \{\mathit{false}\})$ where $\delta_\Pi(\phi_{pre}, a) = \{ \phi_{post} \mid \sem{a}(\phi_{pre}) \subseteq \phi_{post} \}$. We abbreviate $\lang{\Pi_{NFA}}$ as $\lang{\Pi}$. Intuitively, $\lang{\Pi}$ consists of all traces that can be proven infeasible using only assertions in $\Pi$. Thus the following proof rule is sound \cite{FarzanKP13,FarzanKP15,HeizmannHP09}:
\begin{equation}\label{rule:safety}
  \AxiomC{$\exists \Pi \subseteq \assert \ldotp P \subseteq \lang{\Pi}$}
  \UnaryInfC{$P$ is safe}
  \DisplayProof
  \tag{\textsc{Safe}}
\end{equation}

When $P \subseteq \lang{\Pi}$, we say that $\Pi$ is a proof for $P$. A proof does not uniquely belong to any particular program; a single $\Pi$ may prove many programs correct.

%% file: cav19-extended/program.tex

\section{Reductions}
\label{sec:program}

The set of assertions used for a proof is usually determined by a particular language of assertions, and a safe program may not have a (safety) proof in that particular language. Yet, a subset of the program traces may have a proof in that assertion language. If it can be proven that the subset of program runs that have a safety proof are a faithful representation of all program behaviours (with respect to a given property), then the program is correct. This motivates the notion of {\em program reductions}.

\begin{definition}[semantic reduction] \label{def:sr}
  If for programs $P$ and $P'$, $P'$ is safe implies that $P$ is safe, then $P'$ is a \emph{semantic reduction} of $P$ (written $P' \preceq P$).
\end{definition}

The definition immediately gives rise to the following proof rule for proving program safety: 
\begin{equation}\label{rule:safety+reductions-bad}
  \AxiomC{$\exists P' \preceq P, \Pi \subseteq \assert \ldotp P' \subseteq \lang{\Pi}$}
  \UnaryInfC{$P$ is safe}
  \DisplayProof
  \tag{\textsc{SafeRed1}}
\end{equation}
This generic proof rule is not automatable since, given a proof $\Pi$, verifying the existence of the appropriate reduction is {\em undecidable}. Observe that a program is safe if and only if $\emptyset$ is a valid reduction of the program. This means that discovering a semantic reduction and proving safety are mutually reducible to each other. To have decidable premises for the proof rule, we need to formulate an easier (than proving safety) problem in discovering a reduction. One way to achieve this is by  restricting the set of reductions under consideration from all reductions (given in Definition \ref{def:sr}) to a proper subset which more amenable to algorithmic checking.
Fixing a set $\mathcal{R}$ of (semantic) reductions, we will have the rule: 
\begin{equation}\label{rule:safety+reductions}
  \AxiomC{$\exists P' \in \mathcal{R} \ldotp P' \subseteq \lang{\Pi}$}
  \AxiomC{$\forall P' \in \mathcal{R} \ldotp P' \preceq P$}
  \BinaryInfC{$P$ is safe}
  \DisplayProof
  \tag{\textsc{SafeRed2}}
\end{equation}
\begin{proposition}
  \label{prop:proof-rule}
  The proof rule \ref{rule:safety+reductions} is sound.
\end{proposition}

\begin{proof}
By the left precondition, there exists some $P' \in \mathcal{R}$ such that $P' \subseteq \lang{\Pi}$, which implies $\sem{P'} = \emptyset$. By the right precondition, $\sem{P} = \sem{P'}$, and therefore $P$ is safe.
\end{proof}

The core contribution of this paper is that it provides an algorithmic solution inspired by the above proof rule. To achieve this, two subproblems are solved: (1) Given a set $\mathcal{R}$ of reductions of a program $P$ and a candidate proof $\Pi$, can we check if there exists a reduction $P' \in \mathcal{R}$ which is covered by the proof $\Pi$? In section \ref{sec:lta}, we propose a new semantic interpretation of an existing notion of infinite tree automata that gives rise to an algorithmic check for this step.
(2) Given a program $P$, is there a general sound set of reductions $\mathcal{R}$ that be effectively represented to accommodate step (1)? In section \ref{sec:sleep}, we propose a construction of an effective set of reductions, representable by our infinite tree automata, using inspirations from known partial order reduction techniques \cite{Godefroid96}.


%% file: cav19-extended/lta.tex

\section{Proof Checking}
\label{sec:lta}

\begin{wrapfigure}{R}{0.3\textwidth}\vspace{-30pt}
  \begin{center}
  \includegraphics[scale=0.24]{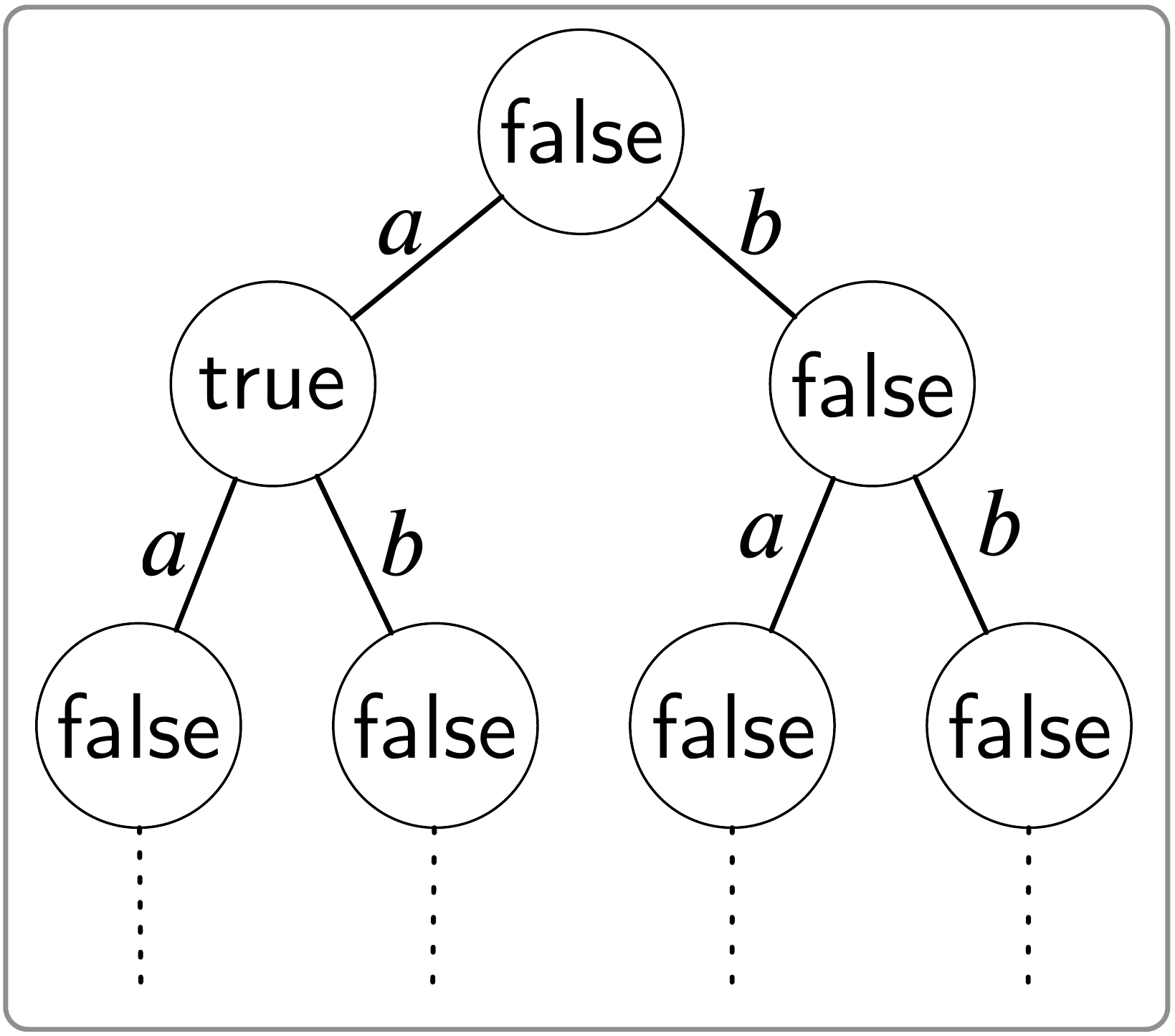}\vspace{-8pt}
  \caption{\small $L = \{a\}$ as an infinite tree.}
  \label{fig:a}\vspace{-15pt}
  \end{center}
\end{wrapfigure}

Given a set of reductions $\mathcal{R}$ of a program $P$, and a candidate proof $\Pi$, we want to check if there exists a reduction $P' \in \mathcal{R}$ which is covered by $\Pi$. We call this \emph{proof checking}. We use tree automata to represent certain classes of languages (i.e sets of sets of strings), and then use operations on these automata for the purpose of proof checking.

The set $\stmt^*$ can be represented as an infinite tree. Each $x \in \stmt^*$ defines a path to a unique node in the tree: the root node is located at the empty string $\epsilon$, and for all $a \in \Sigma$, the node located at $xa$ is a child of the node located at $x$. Each node is then identified by the string labeling the path leading to it. A language $L \subseteq \stmt^*$ (equivalently, $L : \stmt^* \to \bool$) can consequently be represented as an infinite tree where the node at each $x$ is labelled with a boolean value $B \equiv (x \in L)$. An example is given in Figure \ref{fig:a}.
It follows that a set of languages is a set of infinite trees, which can be represented using automata over infinite trees.
Looping Tree Automata (LTAs) are a subclass of B\"{u}chi Tree Automata where all states are accept states \cite{BaaderT01}. The class of Looping Tree Automata is closed under intersection and union, and checking emptiness of LTAs is decidable. Unlike B\"{u}chi Tree Automata, emptiness can be decided in linear time \cite{BaaderT01}.


\begin{definition}
  A Looping Tree Automaton (LTA) over $|\stmt|$-ary, $\bool$-labelled trees is a tuple $M = (Q, \Delta, q_0)$ where $Q$ is a finite set of states, $\Delta \subseteq Q \times \bool \times (\stmt \to Q)$ is the transition relation, and $q_0$ is the initial state.
\end{definition}
Intuitively, an LTA $M = (Q, \Delta, q_0)$ performs a parallel and depth-first traversal of an infinite tree $L$ while maintaining some local state. Execution begins at the root $\epsilon$ from state $q_0$ and non-deterministically picks a transition $(q_0, B, \sigma) \in \Delta$ such that $B$ matches the label at the root of the tree (i.e. $B = (\epsilon \in L)$). If no such transition exists, the tree is rejected. Otherwise, $M$ recursively works on each child $a$ from state $q' = \sigma(a)$ in parallel. This process continues infinitely, and $L$ is accepted if and only if $L$ is never rejected.

Formally, $M$'s execution over a tree $L$ is characterized by a \emph{run} $\delta^* : \stmt^* \to Q$ where $\delta^*(\epsilon) = q_0$ and $(\delta^*(x), x \in L, \lambda a \ldotp \delta^*(xa)) \in \Delta$ for all $x \in \stmt^*$. The set of languages accepted by $M$ is then defined as
  $\lang{M} = \{ L \mid \exists \delta^* \ldotp \text{$\delta^*$ is a run of $M$ on $L$} \}$.

\begin{theorem}\label{thm:pc-decidable}
  Given an LTA $M$ and a regular language $L$, it is decidable whether $
    \exists P \in \lang{M} \ldotp P \subseteq L$.
\end{theorem}

\begin{proof}
  The proposition $\exists P \in \lang{M} \ldotp P \subseteq L$ is equivalent to the proposition $\lang{M} \cap \pow{L} = \emptyset$. LTA languages are closed under intersection (Lemma \ref{lem:lta-intersect}), $\pow{L}$ is recognized by an LTA (Lemma \ref{lem:lta-powerset}), and LTA emptiness is decidable \cite{BaaderT01}, so $\lang{M} \cap \pow{L} = \emptyset$ (and therefore $\exists P \in \lang{M} \ldotp P \subseteq L$) is decidable.
\end{proof}

\begin{lemma} \label{lem:lta-intersect}
  The set of languages accepted by an LTA is closed under intersection.
\end{lemma}

\begin{proof}
  The standard construction for B\"{u}chi tree automata intersection also works for LTAs. There is a simpler construction specifically for LTAs, which we include here.

  Let $M_1 = (Q_1, \Delta_1, q_{01})$ and  $M_2 = (Q_2, \Delta_2, q_{02})$ be LTAs.
  Define $M_\cap = (Q_1 \times Q_2, \Delta_\cap, (q_{01}, q_{02}))$, where
  \[
    \Delta_\cap = \{ ((q_1, q_2), B, \lambda a \ldotp (\sigma_1(a), \sigma_2(a))) \mid (q_1, B, \sigma_1) \in \Delta_1 \land (q_2, B, \sigma_2) \in \Delta_2 \}
  \]
  Then $\lang{M_\cap} = \lang{M_1} \cap \lang{M_2}$.

  This proof has been mechanically checked.
\end{proof}

\begin{lemma}\label{lem:lta-powerset}
  If $L$ is a regular language, then $\pow{L}$ is recognized by an LTA.
\end{lemma}

\begin{proof}
  Since $L$ is a regular language, there exists a DFA $A = (Q, \stmt, \delta, q_0, F)$ such that $\lang{A} = L$. Define $M_\pow{L} = (Q, \Delta_\pow{L}, q_0)$ where
  \[
    \Delta_\pow{L} = \{ (q, B, \lambda a \ldotp \delta(q, a)) \mid B \implies q \in F \}
  \]
  Then $\lang{M_\pow{L}} = \pow{L}$.

  This proof has been mechanically checked.
\end{proof}

\paragraph{\bfseries Counterexamples.} Theorem \ref{thm:pc-decidable} effectively states that proof checking is decidable. For automated verification, beyond checking the validity of a proof, we require counterexamples to fuel the development of the proof when the proof does not check. Note that in the simple case of the proof rule \ref{rule:safety}, when $P \not\subseteq \lang{\Pi}$ there exists a counterexample trace $x \in P$ such that $x \notin \lang{\Pi}$.

With our proof rule \ref{rule:safety+reductions}, things get a bit more complicated. First, note that unlike the classic case (\ref{rule:safety}), where a failed proof check coincides with the non-emptiness of an intersection check (i.e. $P \cap \overline{\lang{\Pi}} \not = \emptyset$), in our case, a failed proof check coincides with the emptiness of an intersection check (i.e. $\mathcal{R} \cap \pow{\lang{\Pi}} = \emptyset$). The sets $\mathcal{R}$ and $\pow{\lang{\Pi}}$ are both sets of languages. What does the witness to the emptiness of the intersection look like? Each language member of $\mathcal{R}$ contains at least one string that does not belong to any of the subsets of our proof language. One can collect all such witness strings to guarantee progress across the board in the next round. However, since LTAs can represent an infinite set of languages, one must take care not end up with an infinite set of counterexamples following this strategy. Fortunately, this will not be the case.

\begin{theorem}
  \label{thm:finite-counterexamples}
  Let $M$ be an LTA and let $L$ be a regular language such that $P \not\subseteq L$ for all $P \in \lang{M}$. There exists a finite set of counterexamples $C$ such that, for all $P \in \lang{M}$, there exists some $x \in C$ such that $x \in P$ and $x \notin L$.
\end{theorem}

\begin{proof}
  Assume $M = (Q_M, \Delta_M, q_{0M})$ and let $A = (Q_A, \stmt, \delta_A, q_{0A}, F_A)$ be an automaton that accepts $L$.

  Assume $P \not\subseteq L$ for all $P \in \lang{M}$. Then $\lang{M} \cap \pow{L} = \emptyset$, so the root node of any automaton accepting $\lang{M} \cap \pow{L}$ is \emph{inactive} \cite{BaaderT01}. The set of inactive states is the smallest set satisfying
  \begin{prooftree}
    \AxiomC{$\forall (q, B, \sigma) \in \Delta \ldotp \exists a \ldotp \sigma(a) \in \inactive(M)$}
    \UnaryInfC{$q \in \inactive(M)$}
  \end{prooftree}
  If we instantiate this rule for $M_\cap$ (from Lemma \ref{lem:lta-intersect}) and $M_\pow{L}$ (from Lemma \ref{lem:lta-powerset}), we get
  \begin{prooftree}
    \AxiomC{$\forall (q_M, B, \sigma) \in \Delta_M \ldotp (B \implies q_A \in F_A) \implies \exists a \ldotp (\sigma(a), \delta_A(q_A, a)) \in \inactive(M_\cap)$}
    \UnaryInfC{$(q_M, q_A) \in \inactive(M_\cap)$}
  \end{prooftree}
  A proof of inactivity is essentially a finite tree where every node is labelled by a pair of states $(q_M, q_A)$ (the root node labelled by $(q_{0M}, q_{0A})$) and contains an outgoing edge labelled by some $a \in \stmt$ for every transition. We define our counterexample set $C$ as the set of all strings labelling the path from the root node to any leaf node in this tree. Such a set is clearly finite, so it remains to show that every $P \in \lang{M}$ contains some element of $C$ that is not in $L$.

  Fix some $P \in \lang{M}$ and let $\delta^* : \stmt^* \to Q_M$ be the corresponding accepting run. Since the every node in the tree described above contains an outgoing edge for every transition, there must exist a sequence of nodes through the tree labelled by a sequence of states $(q_{0M}, q_{0A})...(q_{nM}, q_{nA})$ and string $a = a_1...a_n \in C$ such that $\delta^*(a_1...a_i) = q_{iM}$ for all $0 \le i \le n$. The state $(q_{nM}, q_{nA})$ must have no outgoing transitions or else it would not label a leaf node. However, both $q_{nM}$ and $q_{nA}$ must have outgoing transitions in their own respective automata: $q_{nM}$ is part of an accepting run, and every state in $M_\subseteq$ has an outgoing transition by definition. Thus, $(q_{nM}, B, \sigma) \in \Delta_M$ implies $B = \top$ and $q_{nA} \notin F_A$, and therefore $a \in P$ and $a \notin L$.

  This proof has been mechanically checked.
\end{proof}

This theorem justifies our choice of using LTAs instead of more expressive formalisms such as B\"{u}chi Tree Automata. For example, the B\"{u}chi Tree Automaton that accepts the language $\{ \{x\} \mid x \in \stmt^* \}$ would give rise to an infinite number of counterexamples with respect to the empty proof (i.e. $\Pi = \emptyset$). The finiteness of the counterexample set presents an alternate proof that LTAs are strictly less expressive than B\"{u}chi Tree Automata \cite{VardiW94}.

%% file: cav19-extended/sleep.tex

\section{Sleep Set Reductions}\label{sec:sleep}

We have established so far that (1) a set of assertions gives rise to a regular language proof, and (2) given a regular language proof and a set of program reductions recognizable by an LTA, we can check the program (reductions) against the proof. The last piece of the puzzle is to show that a useful class of program reductions can be expressed using LTAs.

Recall our example from Section 2. The reduction we obtain is sound because, for every trace in the full parallel-composition program, an equivalent trace exists in the reduced program. By equivalent, we mean that one trace can be obtained from the other by swapping independent statements. Such an equivalence is the essence of the theory of Mazurkiewicz traces \cite{DiekertM97}.

We fix a reflexive symmetric \emph{dependence relation} $D \subseteq \stmt \times \stmt$. For all $a, b \in \stmt$, we say that $a$ and $b$ are \emph{dependent} if $(a, b) \in D$, and say they are \emph{independent} otherwise. We define $\sim_D$ as the smallest congruence satisfying $xaby \sim_D xbay$ for all $x, y \in \stmt^*$ and independent $a, b \in \stmt$. The closure of a language $L \subseteq \stmt^*$ with respect to $\sim_D$ is denoted $[L]_D$. A language $L$ is \emph{$\sim_D$-closed} if $L = [L]_D$. It is worthwhile to note that all input programs considered in this paper correspond to regular languages that are $\sim_D$-closed.

An equivalence class of $\sim_D$ is typically called a (Mazurkiewicz) trace. We avoid using this terminology as it conflicts with our definition of traces as strings of statements in Section \ref{sec:ptrace}. We assume $D$ is \emph{sound}, i.e. $\sem{ab} = \sem{ba}$ for all independent $a, b \in \stmt$.

\begin{definition}[$D$-reduction]
  A program $P'$ is a \emph{$D$-reduction} of a program $P$, that is $P' \preceq_D P$,  if $[P']_D = P$.
\end{definition}

Note that the equivalence relation on programs induced by $\sim_D$ is a refinement of the semantic equivalence relation used in Definition \ref{def:sr}.

\begin{lemma}\label{lem:d-reduction}
  If $P' \preceq_D P$ then $P' \preceq P$.
\end{lemma}

\begin{proof}
  Since $D$ is sound, then $x \sim_D y$ implies $\sem{x} = \sem{y}$ for all $x, y \in \stmt^*$. Then
  \begin{align*}
    (a, b) \in \sem{P}
    &\iff \exists x \in P \ldotp (a, b) \in \sem{x}\\
    &\iff \exists x \in [P']_D \ldotp (a, b) \in \sem{x}\\
    &\iff \exists x \ldotp \exists y \in P' \ldotp x \sim_D y \land (a, b) \in \sem{x}\\
    &\iff \exists x \ldotp \exists y \in P' \ldotp x \sim_D y \land (a, b) \in \sem{y}\\
    &\iff \exists y \in P' \ldotp (a, b) \in \sem{y}\\
    &\iff (a, b) \in \sem{P'}
  \end{align*}
  for all $a, b \in \state$, so $\sem{P'} = \sem{P}$ and therefore $P' \preceq P$.
\end{proof}

Ideally, we would like to define an LTA that accepts all $D$-reductions of a program $P$, but unfortunately this is not possible in general.

\begin{proposition}[corollary of Theorem 67 of \cite{DiekertM97}]\label{prop:undecidable-d-reduction}
  For arbitrary regular languages $L_1, L_2 \in \stmt^*$ and relation $D$, the proposition $\exists L \preceq_D L_1 \ldotp L \subseteq L_2$ is undecidable.
\end{proposition}

\begin{proof}
  Assume that we can decide whether $\exists L \preceq_D L_1 \ldotp L \subseteq L_2$. Then we can decide whether $[L']_D = \stmt^*$ for any regular language $L' \subseteq \stmt^*$ (by instantiating $L_1 = \stmt^*$ and $L_2 = L'$), which is known to be generally undecidable \cite{DiekertM97}.
\end{proof}

The proposition is decidable only when $\overline{D}$ is transitive, which does not hold for a semantically correct notion of independence  for a parallel program encoding a $k$-safety property, since statements from the same thread are dependent and statements from different program copies are independent. Therefore, we have:

\begin{proposition}
  Assume $P$ is a $\sim_D$-closed program and $\Pi$ is a proof. The proposition $
    \exists P' \preceq_D P \ldotp P' \subseteq \lang{\Pi}$
  is undecidable.
\end{proposition}

In order to have a decidable premise for proof rule \ref{rule:safety+reductions} then, we present an approximation of the set of $D$-reductions, inspired by sleep sets \cite{Godefroid96}. The idea is to construct an LTA that recognizes {\em a class of $D$-reductions} of an input program $P$, whose language is assumed to be $\sim_D$-closed. This automaton intuitively makes non-deterministic choices about what program traces to {\em prune} in favour of other $\sim_D$-equivalent program traces for a given reduction. Different non-deterministic choices lead to different $D$-reductions.

\begin{wrapfigure}{r}{0.29\textwidth}\vspace{-28pt}
\begin{center}
\caption{Exploring from $x$ with sleep sets.} \label{fig:sse}
\includegraphics[scale=0.26]{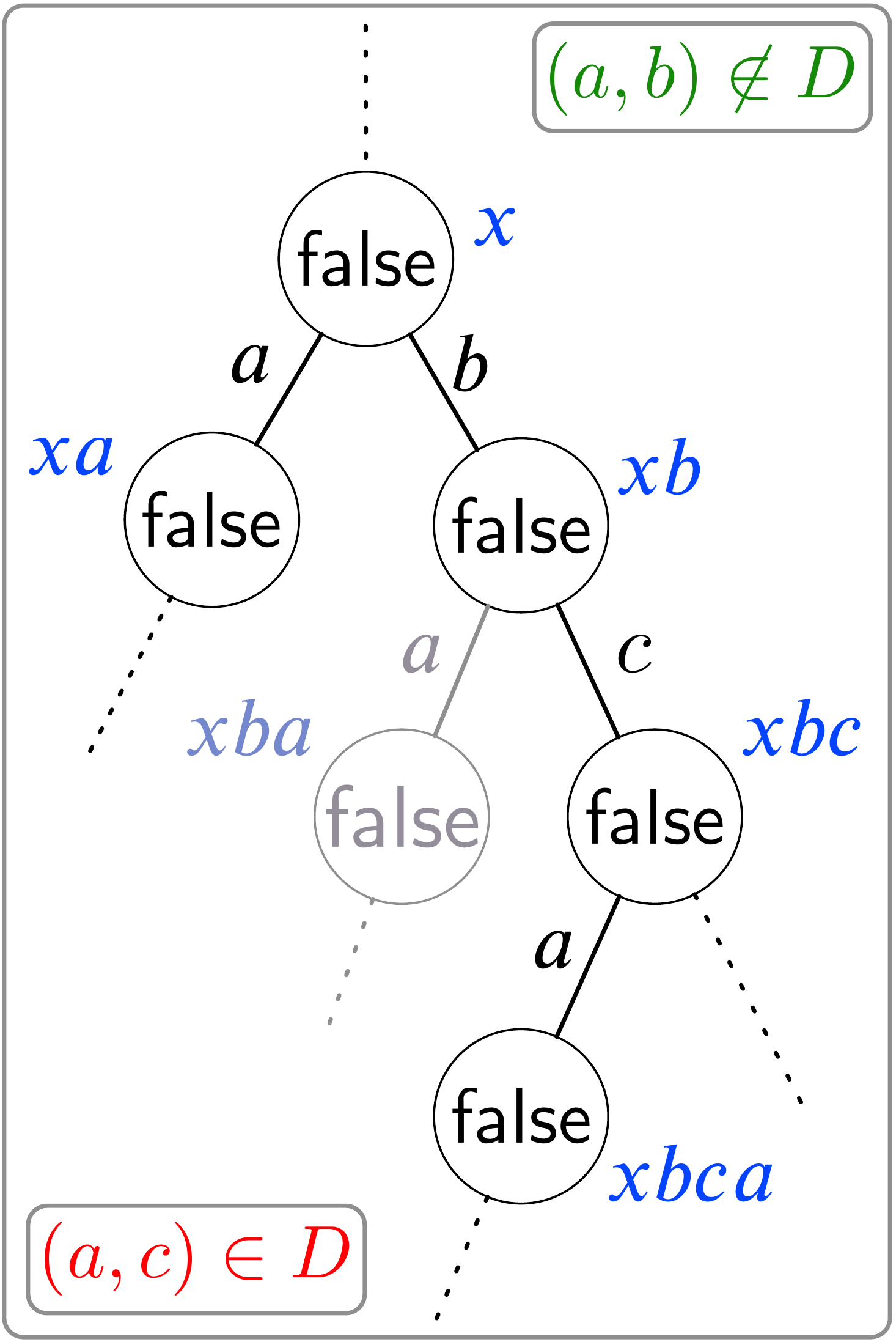}
\end{center}\vspace{-30pt}
\end{wrapfigure}
Consider two statements $a,b \in \stmt$ where $(a,b) \not \in D$. Let $x,y \in \stmt^*$ and consider two program runs $xaby$ and $xbay$. We know $\sem{xbay} = \sem{xaby}$. If the automaton makes a non-deterministic choice that the successors of $xa$ have been explored, then the successors of $xba$ need not be explored (can be pruned away) as illustrated in Figure \ref{fig:sse}.  Now assume $(a,c) \in D$, for some $c \in \stmt$. When the node $xbc$ is being explored, we can no longer safely ignore $a$-transitions, since the equality $\sem{xbcay} = \sem{xabcy}$ is not guaranteed. Therefore, the $a$ successor of $xbc$ has to be explored.
The nondeterministic choice of what child node to explore is modelled by a choice of order in which we explore each node's children. Different orders yield  different reductions. Reductions are therefore characterized as an assignment $R : \stmt^* \to \linear{\stmt}$ from nodes to linear orderings on $\stmt$, where $(a, b) \in R(x)$ means we explore child $xa$ after child $xb$.

Given $R : \stmt^* \to \linear{\stmt}$, the \emph{sleep set} $\sleep_R(x) \subseteq \stmt$ at node $x \in \stmt^*$ defines the set of transitions that can be ignored at $x$:
\begin{align*}
  \sleep_R(\epsilon) &= \emptyset \tag{1} \\
  \sleep_R(xa) &= (\sleep_R(x) \cup R(x)(a)) \setminus D(a) \tag{2}
\end{align*}
Intuitively, (1) no transition can be ignored at the root node, since nothing has been explored yet, and (2) at node $x$, the sleep set of $xa$ is obtained by adding the transitions we explored before $a$ ($R(x)(a)$) and then removing the ones that conflict with $a$ (i.e. are related to $a$ by $D$).

Next, we make precise which nodes are ignored. The set of ignored nodes is the smallest set $\ignore_R : \stmt^* \to \bool$ such that
\begin{align*}
  x \in \ignore_R &\implies xa \in \ignore_R \tag{1} \\
  a \in \sleep_R(x)  &\implies xa \in \ignore_R \tag{2}
\end{align*}
Intuitively, a node $xa$ is ignored if (1) any of its ancestors is ignored ($\ignore_R(x)$), or (2) $a$ is one of the ignored transitions at node $x$ ($a \in \sleep_R(x)$).

Finally, we obtain an actual reduction of a program $P$ from a characterization of a reduction $R$ by removing the ignored nodes from $P$, i.e. $P \setminus \ignore_R$.
\begin{lemma}\label{lem:ignore-reduction}
  For all $R : \stmt^* \to \linear{\stmt}$, if $P$ is a $\sim_D$-closed program then $P \setminus \ignore_R$ is a $D$-reduction of $P$.
\end{lemma}

\begin{proof}
  If $[\overline{\ignore_R}]_D = \stmt^*$, then
  \begin{align*}
    [P \setminus \ignore_R]_D
    &= [P \cap \overline{\ignore_R}]_D \\
    &= P \cap [\overline{\ignore_R}]_D \\
    &= P \cap \stmt^* \\
    &= P
  \end{align*}
  so it is sufficient to show $[\overline{\ignore_R}]_D = \stmt^*$. More specifically, it is sufficient to show that for all $x \in \ignore_R$ there exists some $y \notin \ignore_R$ such that $x \sim_D y$.

  First, observe that
  \[ \ignore_R = \{ x_1ax_2bx_3 \mid (a, b) \in R(x_1) \land (\forall c \in ax_2 \ldotp (c, b) \notin D) \}\]
  Second, define the following ordering on traces
  \[ xay <_R xbz \iff (b, a) \in R(x) \land |y| = |z| \]
  The $|y| = |z|$ condition enforces that only strings of the same length are related. Since each $R(x)$ is a linear order on a finite set, it follows that $<_R$ is well-founded.

  Assume $x \in \ignore_R$. We proceed by well-founded induction on $x$ using $<_R$. By the above observation, we have $x = x_1ax_2bx_3$ for some $x_1, x_2, x_3 \in \stmt$ and $(a, b) \in R(x_1)$ such that $(c, b) \notin D$ for all $c \in ax_2$. Define $y = x_1bax_2x_3$. Thus $x \sim_D y$ and $y <_R x$. If $y \notin \ignore_R$, then we are done. Otherwise, we have $y \in \ignore_R$ and $y <_R x$, and so induction completes the proof.

  This proof has been mechanically checked.
\end{proof}

We define the set of all such reductions as $\reduce_D(P) = \{ P \setminus \ignore_R \mid R : \stmt^* \to \linear{\stmt} \}$.

\begin{theorem}
  \label{thm:sleep-reduction}
  For any regular language $P$,  $\reduce_D(P)$ is accepted by an LTA.
\end{theorem}

\begin{proof}
  Since $P$ is regular, there exists a DFA $A = (Q, \stmt, \delta, q_0, F)$ such that $\lang{A} = P$. Define $M_D = (Q \times \bool \times \pow{\stmt}, \Delta_D, (q_0, \bot, \emptyset))$ where
  \[
    \Delta_D = \{ ((q, \iota, S), q \in F \land \neg \iota, \lambda a \ldotp (\delta(q, a), \iota \lor a \in S, (S \cup R(a)) \setminus D(a)) \mid R \in \linear{\stmt} \}
  \]
  The values of $\sleep_R(x)$ and $x \in \ignore_R$ can be computed by a simple left-to-right traversal of the input string $x$. Intuitively, $M$ simulates this computation for a nondeterministically chosen $R : \stmt^* \to \linear{\stmt}$. Partial computations of $\sleep_R$ ($\ignore_R$) are stored in the $\pow{\stmt}$ ($\bool$) part of the state. Thus $\lang{M_D} = \reduce_D(P)$.

  This proof has been mechanically checked.
\end{proof}

Interestingly, every reduction in $\reduce_D(P)$ is optimal in the sense that each reduction contains at most one representative of each equivalence class of $\sim_D$.

\begin{theorem}
  \label{thm:sleep-optimal}
  Fix some $P \subseteq \stmt^*$ and $R : \stmt^* \to \linear{\stmt}$. For all $(x, y) \in P \setminus \ignore_R$, if $x \sim_D y$ then $x = y$.
\end{theorem}

\begin{proof}
  It is sufficient to show that $\overline{\ignore_R}$ contains at most one representative of each equivalence class of $\sim_D$, i.e. for all $x, y \notin \ignore_R$, if $x \sim_D y$ then $x = y$.

  Assume $x, y \notin \ignore_R$ and $x \sim_D y$. For a contradiction, assume $x \ne y$. Then $x$ and $y$ must differ at some character, so $x = x_1ax_2$ and $y = x_1bx_3$ for $x_i \in \stmt^*$ and $a \ne b$. Assume (by symmetry) that $(a, b) \in R$. We have $x \sim_D y$, so $b$ must appear somewhere in $x_2$ after a run of elements independent of $b$, i.e. $x = x_1ax_{21}bx_{22}$ where $(b, c) \in D$ for every $c \in x_{21}$. However, this implies $b \in \sleep_R(x_1ax_{21})$, which implies $x_1ax_{21}b \in \ignore_R$, and therefore $x \in \ignore_R$, which is a contradiction.

  This proof has been mechanically checked.
\end{proof}


%% file: cav19-extended/algorithm.tex

\section{Algorithms}\label{sec:algorithm}


Figure \ref{fig:loop} illustrates the outline of our verification algorithm. It
is a counterexample-guided abstraction refinement loop in the style of \cite{FarzanKP13,FarzanKP15,HeizmannHP09}. The key
difference is that instead of checking whether some proof $\Pi$ is a proof for the program $P$, it checks if there exists a reduction of the program $P$ that $\Pi$ proves correct.


%

The algorithm relies on an oracle \textsc{Interpolate} that, given a finite set of program traces $C$, returns a proof $\Pi'$, if one exists, such that $C \subseteq \lang{\Pi'}$.  In our tool, we use Craig interpolation to implement the oracle \textsc{Interpolate}. In general, since program traces are the simplest form of sequential programs (loop and branch free), any automated program prover that can handle proving them may be used.

\begin{wrapfigure}{r}{0.625\textwidth}\vspace{-20pt}
\hspace{5pt}\includegraphics[scale=0.19]{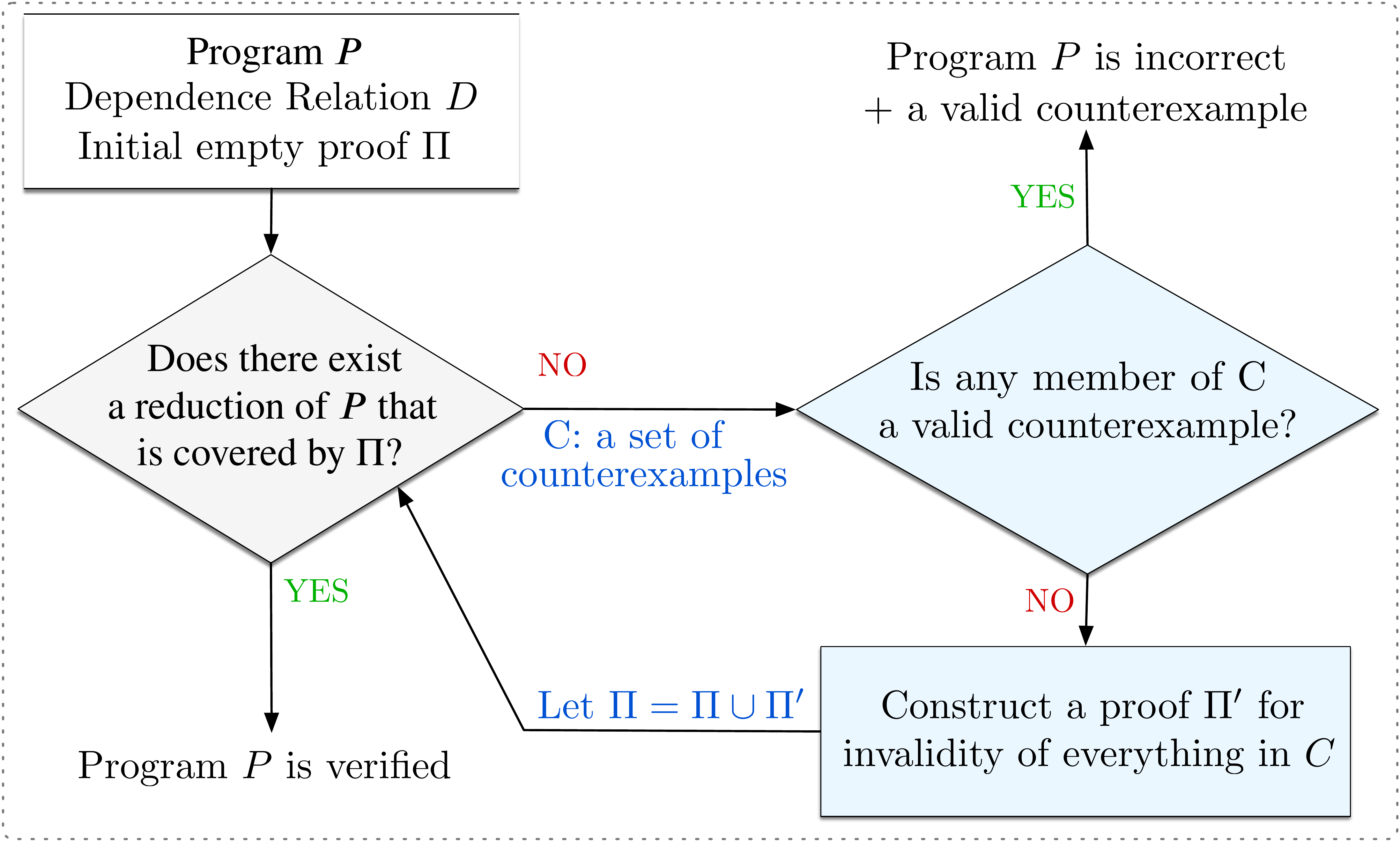}\vspace{-5pt}
\caption{Counterexample-guided refinement loop.}\vspace{-25pt}
\label{fig:loop}
\end{wrapfigure}
The results presented in  Sections \ref{sec:lta} and \ref{sec:sleep} give rise to the proof checking sub routine of the algorithm in Figure \ref{fig:loop} (i.e. the light grey test). Given a program DFA $A_P = (Q_P, \stmt, \delta_P, q_{P0}, F_P)$ and a proof DFA $A_\Pi = (Q_\Pi, \stmt, \delta_\Pi, q_{\Pi0}, F_\Pi)$ (obtained by determinizing the proof NFA $\Pi_{NFA}$), we can decide $\exists P' \in \reduce_D(\lang{A_P}) \ldotp P' \subseteq \lang{A_\Pi}$ by constructing an LTA $M_{P\Pi}$ for $\reduce_D(\lang{A_P}) \cap \pow{\lang{A_\Pi}}$ and checking emptiness (Theorem \ref{thm:pc-decidable}).

\subsection{Progress}

The algorithm corresponding to Figure \ref{fig:loop} satisfies a weak progress theorem: none of the counterexamples from a round of the algorithm will ever appear in a future counterexample set. This, however, is not strong enough to guarantee termination. Alternatively, one can think of the algorithm's progress as follows. In each round new assertions are discovered through the oracle \textsc{Interpolate}, and one can optimistically hope that one can finally converge on an existing target proof $\Pi^*$. The success of this algorithm depends on two factors: (1) the counterexamples used by the algorithm belong to $\lang{\Pi^*}$ and (2) the proof that  \textsc{Interpolate} discovers for these counterexamples coincide with $\Pi^*$. The latter is a typical known wild card in software model checking, which cannot be guaranteed; there is plenty of empirical evidence, however, that procedures based on Craig Interpolation do well in approximating it. The former is a new problem for our refinement loop.

In a standard algorithm in the style of \cite{FarzanKP13,FarzanKP15,HeizmannHP09}, the verification proof rule dictates that every program trace must be in $\lang{\Pi^*}$. In our setting, we only require a subset (corresponding to some reduction) to be in $\lang{\Pi^*}$. This means one cannot simply rely on program traces as {\em appropriate} counterexamples. Theorem \ref{thm:finite-counterexamples} presents a solution to this problem. It ensures that we always feed \textsc{Interpolate} some counterexample from $\Pi^*$ and therefore guarantee progress.

\begin{theorem}[Strong Progress]\label{thm:stp}
  Assume a proof $\Pi^*$ exists for some reduction $P^* \in \mathcal{R}$ and \emph{\textsc{Interpolate}} always returns some subset of $\Pi^*$ for traces in $\lang{\Pi^*}$. Then the algorithm will terminate in at most $|\Pi^*|$ iterations.
\end{theorem}

\begin{proof}
  It is sufficient to show that we learn at least one new assertion in $\Pi^*$ every iteration. Assume we have received a counterexample set $C$ such that, for all $P' \in \mathcal{R}$, there exists some $x \in C$ such that $x \in P'$ and $x \notin \lang{\Pi}$ (Theorem \ref{thm:finite-counterexamples} ensures $C$ exists). Let $x^* \in C$ be the counterexample for $P^*$. Then \textsc{Interpolate}$(x)$ will return new assertions $\Pi' \subseteq \Pi^*$ satisfying $x^* \in \lang{\Pi'}$. If $\Pi' \subseteq \Pi$ then $x^*$ would not have been returned as a counterexample, so there must exist some $\phi \in \Pi'$ (and therefore $\phi \in \Pi^*$) such that $\phi \notin \Pi$.
\end{proof}

Theorem \ref{thm:stp} ensures that the algorithm will never get into an infinite loop due to a bad choice of counterexamples. The condition on \textsc{Interpolate} ensures that divergence does not occur due to the wrong choice of assertions by \textsc{Interpolate} and without it any standard interpolation-based software model checking algorithm may diverge. The assumption that there exists a proof for a reduction of the program in the fixed set $\mathcal{R}$ ensures that the proof checking procedure can verify the target proof $\Pi^*$ once it is reached. Note that, in general, a proof may exist for a reduction of the program which is not in $\mathcal{R}$. Therefore, the algorithm is not complete with respect to all reductions, since checking the premises of \ref{rule:safety+reductions-bad}  is undecidable as discussed in Section \ref{sec:program}.

\subsection{Faster Proof Checking through Antichains}
\label{sec:algorithm:antichains}

The state set of $M_{P\Pi}$, the intersection of program and proof LTAs, has size $|Q_P \times \bool \times \mathcal{P}(\stmt) \times Q_\Pi|$, which is exponential in $|\stmt|$. Therefore, even a linear emptiness test for this LTA can be computationally expensive.  Antichains have been previously used \cite{WulfDHR06}  to optimize certain operations over NFAs that also suffer from exponential blowups, such as deciding universality and inclusion tests. The main idea is that these operations involve computing downwards-closed and upwards-closed sets according to an appropriate subsumption relation, which can be represented compactly as antichains. We employ similar techniques to propose a new emptiness check algorithm.

\paragraph{\bfseries Antichains.} The set of maximal elements of a set $X$ with respect to some ordering relation $\sqsubseteq$ is denoted $\operatorname{max}(X)$. The downwards-closure of a set $X$ with respect to $\sqsubseteq$ is denoted $\floor{X}$. An antichain is a set $X$ where no element of $X$ is related (by $\sqsubseteq$) to another. The maximal elements $\operatorname{max}(X)$ of a finite set $X$ is an antichain. If $X$ is downwards-closed then $\floor{\operatorname{max}(X)} = X$.

\par\noindent\newline The emptiness check algorithm for LTAs from \cite{BaaderT01} computes the set of {\em inactive} states (i.e. states which generate an empty language) and checks if the initial state is inactive. The set of inactive states of an LTA $M = (Q, \Delta, q_0)$ is defined as the smallest set $\inactive(M)$ satisfying
\begin{equation}
  \label{rule:inactive}
  \AxiomC{$\forall (q, B, \sigma) \in \Delta \ldotp \exists a \ldotp \sigma(a) \in \inactive(M)$}
  \UnaryInfC{$q \in \inactive(M)$}
  \DisplayProof
  \tag{\textsc{Inactive}}
\end{equation}
Alternatively, one can view $\inactive(M)$ as the least fixed-point of a monotone (with respect to $\subseteq$) function $F_M : \pow{Q} \to \pow{Q}$
where
\[
  F_M(X) = \{ q \mid \forall (q, B, \sigma) \in \Delta \ldotp \exists a \ldotp \sigma(a) \in X \}.
\]
Therefore, $\inactive(M)$ can be computed using a standard fixpoint algorithm.

If $\inactive(M)$ is downwards-closed with respect to some \textit{subsumption relation} $(\sqsubseteq) \subseteq Q \times Q$, then we need not represent all of $\inactive(M)$. The antichain $\max(\inactive(M))$ of maximal elements of $\inactive(M)$ (with respect to $\sqsubseteq$) would be sufficient to represent the entirety of $\inactive(M)$, and can be exponentially smaller than $\inactive(M)$, depending on the choice of relation $\sqsubseteq$.

A trivial way to compute $\max(\inactive(M))$ is to first compute $\inactive(M)$ and then find the maximal elements of the result, but this involves doing strictly more work than the baseline algorithm. However, observe that if $F_M$ also preserves downwards-closedness with respect to $\sqsubseteq$, then
\begin{align*}
  \max(\inactive(M))
  =& \max(\operatorname{lfp}(F_M)) \\
  =& \max(\operatorname{lfp}(F_M \circ \floor{-} \circ \max))
  = \operatorname{lfp}(\max \circ F_M \circ \floor{-})
\end{align*}
That is, $\max(\inactive(M))$ is the least fixed-point of a function $F^{\max}_M : \pow{Q} \to \pow{Q}$ defined as $F^{\max}_M(X) = \max(F_M(\floor{X}))$. We can calculate $\max(\inactive(M))$ efficiently if we can calculate $F^{\max}_M(X)$ efficiently, which is true in the special case of the intersection automaton for the languages of our proof $\pow{\lang{\Pi}}$ and our program $\reduce_D(P)$, which we refer to as $M_{P\Pi}$.

We are most interested in the state space of $M_{P\Pi}$, which is $Q_{P\Pi} = (Q_P \times \bool \times \pow{\stmt}) \times Q_\Pi$. Observe that states whose $\bool$ part is $\top$ are always active:
\begin{lemma}\label{lem:keep-active}
  $((q_P, \top, S), q_\Pi) \notin \inactive(M_{P\Pi})$ for all $q_P \in Q_P$, $q_\Pi \in Q_\Pi$, and $S \subseteq \stmt$.
\end{lemma}

\begin{proof}
  Let $A_P = (Q_P, \stmt, \delta_P, q_{0P}, F_P)$ and $A_\Pi = (Q_\Pi, \stmt, \delta_\Pi, q_{0\Pi}, F_\Pi)$ be automata recognizing $P$ and $\lang{\Pi}$, respectively. Then
  \begingroup
  \allowdisplaybreaks
  \begin{align*}
     &\ F_{M_{P\Pi}}(X) \\
    =&\ \{ q \mid \forall (q, B, \sigma) \in \Delta_{P\Pi} \ldotp \exists a \ldotp \sigma(a) \in X \} \\
    =&\ \{ q \mid \forall (q, B, \sigma) \in \Delta_\cap \ldotp \exists a \ldotp \sigma(a) \in X \}
      \tag{$M_{P\Pi}$ is an intersection construction}\\
    =&\ \{ \begin{aligned}[t] & (q_P, q_\Pi) \mid \\
        & \forall (q_P, B, \sigma_1) \in \Delta_D, (q_\Pi, B, \sigma_2) \in \Delta_\pow{\lang{\Pi}} \ldotp \\
        & \exists a \ldotp (\sigma_1(a), \sigma_2(a)) \in X \}
    \end{aligned} \tag{Expanding $\Delta_\cap$ from Lemma \ref{lem:lta-intersect}}\\
    =&\ \{ \begin{aligned}[t] & (q_P, q_\Pi) \mid \\
        & \forall (q_P, B, \sigma_1) \in \Delta_D \ldotp (B \implies q_\Pi \in F_\Pi) \implies \\
        & \exists a \ldotp (\sigma_1(a), \delta_\Pi(q_\Pi, a)) \in X \}
    \end{aligned} \tag{Expanding $\Delta_\pow{\lang{\Pi}}$ from Lemma \ref{lem:lta-powerset}}\\
    =&\ \{ \begin{aligned}[t] & ((q_P, \iota, S), q_\Pi) \mid \\
         & \forall R \in \linear{R} \ldotp (q_P \in F_P \land \neg \iota \implies q_\Pi \in F_\Pi) \implies  \\
         & \exists a \ldotp ((q_P', \iota \lor a \in S, (S \cup R(a)) \setminus D(a)), \delta_\Pi(q_\Pi, a)) \in X \}
    \end{aligned} \tag{Expanding $\Delta_D$ from Lemma \ref{thm:sleep-reduction}}
  \end{align*}
  \endgroup
  Note that if $\iota = \top$ then the body of the set comprehension simplifies to
  \[ \forall R \in \linear{R} \ldotp \exists a \ldotp ((q_P', \top, (S \cup R(a)) \setminus D(a)), \delta_\Pi(q_\Pi, a)) \in X \]
  In other words, a state where $\iota = \top$ always has a transition to another state with $\iota = \top$. Therefore such states cannot be inactive.

  This proof has been mechanically checked.
\end{proof}

The state space can then be assumed to be $Q_{P\Pi} = (Q_P \times \{\bot\} \times \pow{\stmt}) \times Q_\Pi$ for the purposes of checking inactivity. The subsumption relation defined as the smallest relation $\sqsubseteq_{P\Pi}$ satisfying
\[
  S \subseteq S' \implies ((q_P, \bot, S), q_\Pi) \sqsubseteq_{P\Pi} ((q_P, \bot, S'), q_\Pi)
\]
for all $q_P \in Q_P$, $q_\Pi \in Q_\Pi$, and $S, S' \subseteq \stmt$, is a suitable one since:
\begin{lemma}\label{lem:preserve-dc}
  $F_{M_{P\Pi}}$ preserves downwards-closedness with respect to $\sqsubseteq_{P\Pi}$.
\end{lemma}

\begin{proof}
  Let $A_P = (Q_P, \stmt, \delta_P, q_{0P}, F_P)$ and $A_\Pi = (Q_\Pi, \stmt, \delta_\Pi, q_{0\Pi}, F_\Pi)$ be automata recognizing $P$ and $\lang{\Pi}$, respectively. Then
  \begingroup
  \allowdisplaybreaks
  \begin{align*}
     &\ F_{M_{P\Pi}}(X) \\
    =&\ \{ q \mid \forall (q, B, \sigma) \in \Delta_{P\Pi} \ldotp \exists a \ldotp \sigma(a) \in X \} \\
    =&\ \{ q \mid \forall (q, B, \sigma) \in \Delta_\cap \ldotp \exists a \ldotp \sigma(a) \in X \}
      \tag{$M_{P\Pi}$ is an intersection construction}\\
    =&\ \{ \begin{aligned}[t] & (q_P, q_\Pi) \mid \\
        & \forall (q_P, B, \sigma_1) \in \Delta_D, (q_\Pi, B, \sigma_2) \in \Delta_\pow{\lang{\Pi}} \ldotp \\
        & \exists a \ldotp (\sigma_1(a), \sigma_2(a)) \in X \}
    \end{aligned} \tag{Expanding $\Delta_\cap$ from Lemma \ref{lem:lta-intersect}}\\
    =&\ \{ \begin{aligned}[t] & (q_P, q_\Pi) \mid \\
        & \forall (q_P, B, \sigma_1) \in \Delta_D \ldotp (B \implies q_\Pi \in F_\Pi) \implies \\
        & \exists a \ldotp (\sigma_1(a), \delta_\Pi(q_\Pi, a)) \in X \}
    \end{aligned} \tag{Expanding $\Delta_\pow{\lang{\Pi}}$ from Lemma \ref{lem:lta-powerset}}\\
    =&\ \{ \begin{aligned}[t] & ((q_P, \iota, S), q_\Pi) \mid \\
         & \forall R \in \linear{R} \ldotp (q_P \in F_P \land \neg \iota \implies q_\Pi \in F_\Pi) \implies  \\
         & \exists a \ldotp ((q_P', \iota \lor a \in S, (S \cup R(a)) \setminus D(a)), \delta_\Pi(q_\Pi, a)) \in X \}
    \end{aligned} \tag{Expanding $\Delta_D$ from Lemma \ref{thm:sleep-reduction}} \\
    =&\ \{ \begin{aligned}[t] & ((q_P, \bot, S), q_\Pi) \mid \\
         & \forall R \in \linear{R} \ldotp (q_P \in F_P \implies q_\Pi \in F_\Pi) \implies  \\
         & \exists a \notin S \ldotp ((q_P', \bot, (S \cup R(a)) \setminus D(a)), \delta_\Pi(q_\Pi, a)) \in X \}
    \end{aligned} \tag{Restricting $\bool$-part of the domain to $\{\bot\}$} \\
  \end{align*}
  \endgroup
  where
  \begin{align*}
    q_P' &= \delta_P(q_P, a) & q_\Pi' &= \delta_\Pi(q_\Pi, a)
  \end{align*}
  Recall the subsumption relation $\sqsubseteq_{P\Pi}$:
  \[
    S \subseteq S' \implies ((q_P, \bot, S), q_\Pi) \sqsubseteq ((q_P, \bot', S'), q_\Pi)
  \]
  Assume $X \subseteq (Q_P \times \{\bot\} \times \pow{\stmt}) \times Q_\Pi$ is downwards-closed with respect to $\sqsubseteq_{P\Pi}$. For any $S, S' \subseteq \stmt$ such that $S \subseteq S'$, we have
  \[
    (S \cup R(a)) \setminus D(a) \subseteq S' \cup R(a) \setminus D(a)
  \]
  for all $R \in \linear{\stmt}$, so $F_{P\Pi}(X)$ is also downwards-closed with respect to $\sqsubseteq_{P\Pi}$.

  This proof has been mechanically checked.
\end{proof}

The function $F^{\max}_{M_{P\Pi}}$ is a function over relations
\[
  F^{\max}_{M_{P\Pi}} : \pow{(Q_P \times \{\bot\} \times \pow{\stmt}) \times Q_\Pi} \to \pow{(Q_P \times \{\bot\} \times \pow{\stmt}) \times Q_\Pi}
\]
but in our case it is more convenient to view it as a function over functions
\[
  F^{\max}_{M_{P\Pi}} : (Q_P \times \{\bot\} \times Q_\Pi \to \pow{\pow{\stmt}}) \to (Q_P \times \{\bot\} \times Q_\Pi \to \pow{\pow{\stmt}})
\]

Through some algebraic manipulation and some simple observations, we can define $F^{\max}_{M_{P\Pi}}$ functionally as follows.
\begin{lemma}\label{lem:fmax-def}
  For all $q_P \in Q_P$, $q_\Pi \in Q_\Pi$, and $X : Q_P \times \{\bot\} \times Q_\Pi \to \pow{\pow{\stmt}}$,
\[
  F^{\max}_{M_{P\Pi}}(X)(q_P, \bot, q_\Pi) = \begin{cases}
    \{ \stmt \} & \text{if $q_P \in F_P \land q_\Pi \notin F_\Pi$} \\
    \bigsqcap \limits_{R \in \linear{\stmt}}
    \bigsqcup \limits_{\substack{
      a \in \Sigma \\
      S \in X(q'_P, \bot, q'_\Pi)}}
      S'
    & \text{otherwise}
\end{cases}
\]
where
\begin{align*}
  q'_P &= \delta_P(q_P, a) &
  X \sqcap Y &= \operatorname{max} \{ x \cap y \mid x \in X \land y \in Y \} \\
  q'_\Pi &= \delta_\Pi(q_\Pi, a) &
  X \sqcup Y &= \operatorname{max} (X \cup Y)
\end{align*}
\[
  S' = \begin{cases}
    \{ (S \cup D(a)) \setminus \{a \} \} & \text{if $R(a) \setminus D(a) \subseteq S$} \\
    \emptyset & \text{otherwise}
  \end{cases}
\]
\end{lemma}

\begin{proof}
  From the proof of Lemma \ref{lem:preserve-dc} we have
  \begin{align*}
    &\ F_{M_{P\Pi}}(X) \\
   =&\ \{ \begin{aligned}[t] & ((q_P, \bot, S), q_\Pi) \mid \\
        & \forall R \in \linear{R} \ldotp (q_P \in F_P \implies q_\Pi \in F_\Pi) \implies  \\
        & \exists a \notin S \ldotp ((q_P', \bot, (S \cup R(a)) \setminus D(a)), q_\Pi') \in X \}
          \end{aligned}
  \end{align*}
  where
  \begin{align*}
    q_P' &= \delta_P(q_P, a) & q_\Pi' &= \delta_\Pi(q_\Pi, a)
  \end{align*}
  Since $\pow{(Q_\Pi \times \{\bot\} \times \pow{\stmt}) \times Q_P} \simeq Q_\Pi \times \{\bot\} \times Q_P \to \pow{\pow{\stmt}}$ we can reformulate $F_{M_{P\Pi}}$ as a function
  \begin{align*}
    &\ F_{M_{P\Pi}}(X)(q_P, \bot, q_\Pi) \\
   =&\ \{ \begin{aligned}[t] S \mid
        & \forall R \in \linear{R} \ldotp (q_P \in F_P \implies q_\Pi \in F_\Pi) \implies  \\
        & \exists a \notin S \ldotp (S \cup R(a)) \setminus D(a) \in X(q'_P, \bot, q_\Pi') \}
          \end{aligned}
  \end{align*}
  and therefore
  \begin{align*}
    &\ F^{\max}_{M_{P\Pi}}(X)(q_P, \bot, q_\Pi) \\
   =&\ \max \{ \begin{aligned}[t] S \mid
        & \forall R \in \linear{R} \ldotp (q_P \in F_P \implies q_\Pi \in F_\Pi) \implies  \\
        & \exists a \notin S \ldotp (S \cup R(a)) \setminus D(a) \in \lfloor X(q'_P, \bot, q_\Pi') \rfloor \}
          \end{aligned} \\
   =&\ \begin{cases}
        \{\Sigma\} & \text{if $q_P \in F_P \land q_\Pi \notin F_\Pi$} \\
        \max \{ \begin{aligned}[t] S \mid
          & \forall R \in \linear{R} \ldotp \exists a \notin S \ldotp \\
          & (S \cup R(a)) \setminus D(a) \in \lfloor X(q'_P, \bot, q_\Pi') \rfloor \}
          \end{aligned} & \text{otherwise}
      \end{cases}
  \end{align*}
  The first case is already in the form we want, so we focus on the second case:
  \begingroup
  \allowdisplaybreaks
  \begin{align*}
    &\ \max \{ \begin{aligned}[t] S \mid
          & \forall R \in \linear{R} \ldotp \exists a \notin S \ldotp \\
          & (S \cup R(a)) \setminus D(a) \in \lfloor X(q'_P, \bot, q_\Pi') \rfloor \}
          \end{aligned} \\
   =&\ \max \{ \begin{aligned}[t] S \mid
          & \forall R \in \linear{R} \ldotp \exists a \in \Sigma, S^\uparrow \in X(q'_P, \bot, q_\Pi') \ldotp \\
          & a \notin S \land (S \cup R(a)) \setminus D(a) \subseteq S^\uparrow \}
          \end{aligned} \\
   =&\  \bigsqcap\limits_{R \in \linear{R}}
        \bigsqcup \limits_{\substack{
          a \in \Sigma \\
          S^\uparrow \in X(q'_P, \bot, q'_\Pi)}}
        \max \{ a \notin S \land (S \cup R(a)) \setminus D(a) \subseteq S^\uparrow \} \\
   =&\  \bigsqcap\limits_{R \in \linear{R}}
        \bigsqcup \limits_{\substack{
          a \in \Sigma \\
          S^\uparrow \in X(q'_P, \bot, q'_\Pi)}}
        \max \{ a \notin S \land S \cup R(a) \subseteq S^\uparrow \cup D(a) \} \\
   =&\  \bigsqcap\limits_{R \in \linear{R}}
        \bigsqcup \limits_{\substack{
          a \in \Sigma \\
          S^\uparrow \in X(q'_P, \bot, q'_\Pi)}}
        \max \{ a \notin S \land S \subseteq S^\uparrow \cup D(a) \land R(a) \subseteq S^\uparrow \cup D(a) \} \\
   =&\  \bigsqcap\limits_{R \in \linear{R}}
        \bigsqcup \limits_{\substack{
          a \in \Sigma \\
          S^\uparrow \in X(q'_P, \bot, q'_\Pi)}}
        \max \{ S \subseteq (S^\uparrow \cup D(a)) \setminus \{a\} \land R(a) \subseteq S^\uparrow \cup D(a) \} \\
   =&\  \bigsqcap\limits_{R \in \linear{R}}
        \bigsqcup \limits_{\substack{
          a \in \Sigma \\
          S \in X(q'_P, \bot, q'_\Pi)}}
        S' \\
  \end{align*}
  \endgroup

  This proof has been mechanically checked.
\end{proof}
Formulating $F^{\max}_{M_{P\Pi}}$ as a higher-order function allows us to calculate $\max(\inactive(M_{P\Pi}))$ using efficient fixpoint algorithms like the one in \cite{Pottier_lazyleast}. Algorithm \ref{alg:total} outlines our proof checking routine. $\textsc{Fix} : ((A \to B) \to (A \to B)) \to (A \to B)$ is a procedure that computes the least fixpoint of its input. The algorithm simply computes the fixpoint of the function $F^{\max}_{M_{P\Pi}}$ as defined in Lemma \ref{lem:fmax-def}, which is a compact representation of $\inactive(M_{P\Pi})$ and checks if the start state of $M_{P\Pi}$ is in it.

\begin{algorithm}[h]
  \DontPrintSemicolon
  \SetAlgoLined
  \SetKwFor{ForAll}{for}{}{}
  \SetKwFunction{Check}{Check}
  \SetKwFunction{FMax}{FMax}
  \SetKwFunction{Fix}{Fix}
  \SetKwFunction{X}{X}
  \SetKwIF{If}{ElseIf}{Else}{if}{}{elif}{else:}{}
  \SetKwProg{Function}{function}{}{}
  \Function{\Check{$A_P$, $A_\Pi$, $D$}}{
    $(Q_P, \Sigma, \delta_P, q_{0P}, F_P) \gets A_P$\;
    $(Q_\Pi, \Sigma, \delta_\Pi, q_{0\Pi}, F_\Pi) \gets A_\Pi$\;
    \Function{\FMax{\FuncSty{X)(}$(q_P, \bot, q_\Pi)$}}{
      \If{$q_P \in F_P \land q_\Pi \notin F_\Pi$}{
        \Return{$\{\stmt\}$}\;}

      $X^\sqcap \gets \{\stmt\}$\;
      \ForAll{$R \in \linear{\stmt}$}{
        $X^\sqcup \gets \emptyset$\;
        \ForAll{$a \in \stmt$, $S \in$ \X{$(\delta_P(q_P, a), \bot, \delta_\Pi(q_\Pi, a))$}}{
          \If{$R(a) \setminus D(a) \subseteq S$}{
            $X^\sqcup \gets X^\sqcup \sqcup \{ (S \cup D(a)) \setminus \{ a \} \}$\;}}
        $X^\sqcap \gets X^\sqcap \sqcap X^\sqcup$\;}
      \Return{$X^\sqcap$}\;}
    \Return{\Fix{\FuncSty{FMax)(}\ArgSty{$(q_{0P}, \bot, q_{0\Pi})$}} $\ne \emptyset$}\;}
   \caption{Proof checking algorithm}
  \label{alg:total}
\end{algorithm}

\paragraph{\bfseries Complexity.} Antichain methods do not generally improve worst case time complexity, as the size of the largest antichain in $\pow{\stmt}$ is exponential in $|\stmt|$. Therefore, our fixpoint algorithm can perform up to $|Q_P| |Q_\Pi| 2^{|\stmt|}$ iterations in the worst case. At each iteration of the fixpoint algorithm,  $F^{\max}(X)(q_P, \bot, q_\Pi)$ must be recalculated for each $q_P \in Q_P$ and $q_\Pi \in Q_\Pi$, where $X : Q_P \times \bot \times Q_\Pi \to \pow{\pow{\stmt}}$ is the current assignment calculated by the fixpoint algorithm.

To analyze the complexity of $F^{\max}$, first, note that antichain meet ($\sqcap$) and join ($\sqcup$) can be computed in $\mathcal{O}((n_1n_2)^2)$ and $\mathcal{O}(n_1n_2)$ time, respectively, where $n_1$ and $n_2$ are the cardinalities of the left and right arguments, respectively. The complexities of iterated antichain meet ($\bigsqcap$) and join ($\bigsqcup$) over a set of $n$ elements with cardinality at most $m$ are therefore
\[ \mathcal{O}((mm)^2 + (m^2m)^2 + \cdots + (m^{n - 1}m)^2) = \mathcal{O}\left(m^2\frac{1 - m^{2n}}{1 - m^2}\right) = \mathcal{O}(m^{2n}) \]
and
\[ \mathcal{O}(mm + (2m)m + \cdots + ((n - 1)m)m) = \mathcal{O}(n^2m^2) \]
respectively.

The inner join of $F^{\max}$ is over at most $k = |\stmt| |Q_P| |Q_\Pi| 2^{|\stmt|}$ antichains of size $\le 1$, and is therefore computed in $\mathcal{O}(k^2)$ and computes an antichain with size $k$. The outer meet is over $|\stmt|!$ antichains produced by the inner join, and thus has complexity $\mathcal{O}(k^{2|\stmt|!})$. Therefore $F^{\max}$ can be computed in $\mathcal{O}(k^{2|\stmt|!}k^2) = \mathcal{O}(k^{2|\stmt|! + 2})$. This brings the total complexity of our fixpoint algorithm to $\mathcal{O}(|Q_P|^2 |Q_\Pi|^2 2^{|\stmt|} k^{2|\stmt|! + 2})$. Since $|\Sigma|$ is typically linear in the program size, we simplify this to $\mathcal{O}(|P| |\Pi| (|P|^2 |\Pi| 2^{|P|})^{2|P|! + 3})$, where $|P|$ is the program size and $|\Pi|$ is the proof size.

\paragraph{\bfseries Counterexamples.}
Theorem \ref{thm:finite-counterexamples} states that a finite set of counterexamples exists whenever $\exists P' \in \reduce_D(P) \ldotp P' \subseteq \lang{\Pi}$ does not hold. The proof of emptiness for an LTA, formed using rule \ref{rule:inactive} above, is a finite tree. Each edge in the tree is labelled by an element of $\stmt$ (obtained from the existential in the rule) and the paths through this tree form the counterexample set. To compute this set, then, it suffices to remember enough information during the computation of $\inactive(M)$ to reconstruct the proof tree. Every time a state $q$ is determined to be inactive, we must also record the witness $a \in \stmt$ for each transition $(q, B, \sigma) \in \Delta$ such that $\sigma(a) \in \inactive(M)$.

In an antichain-based algorithm, once we determine a state $q$ to be inactive, we simultaneously determine everything it subsumes (i.e. $\sqsubseteq q$) to be inactive as well. If we record unique witnesses for each and every state that $q$ subsumes, then the space complexity of our antichain algorithm will be the same as the unoptimized version. The following lemma states that it is sufficient to record witnesses only for $q$ and discard witnesses for states that $q$ subsumes.

\begin{lemma}\label{lem:trans-bij}
  Fix some states $q, q'$ such that $q' \sqsubseteq_{P\Pi} q$. A witness used to prove $q$ is inactive can also be used to prove $q'$ is inactive.
\end{lemma}

\begin{proof}
  From the proof of Lemma \ref{lem:preserve-dc} we have
  \begin{align*}
    &\ F_{M_{P\Pi}}(X) \\
   =&\ \{ \begin{aligned}[t] & ((q_P, \bot, S), q_\Pi) \mid \\
        & \forall R \in \linear{R} \ldotp (q_P \in F_P \implies q_\Pi \in F_\Pi) \implies  \\
        & \exists a \notin S \ldotp ((q_P', \bot, (S \cup R(a)) \setminus D(a)), q_\Pi') \in X \}
          \end{aligned}
  \end{align*}
  where
  \begin{align*}
    q_P' &= \delta_P(q_P, a) & q_\Pi' &= \delta_\Pi(q_\Pi, a)
  \end{align*}
  Thus, if we have states $q' = ((q_P, \bot, S'), q_\Pi)$ and $q = ((q_P, \bot, S), q_\Pi)$ with $S \subseteq S'$ such that $q' \in \inactive(M_{P\Pi})$ (and therefore $q \in \inactive(M_{P\Pi})$), our witness for $q'$ is a function $f : \linear{R} \to \stmt$ such that
  \begin{align*}
    & \forall R \in \linear{R} \ldotp (q_P \in F_P \implies q_\Pi \in F_\Pi) \implies \\
    & f(R) \notin S' \land ((q_P', \bot, (S' \cup R(a)) \setminus D(a)), q_\Pi') \in \inactive(M_{P\Pi})
  \end{align*}
  We must show
  \begin{align*}
    & \forall R \in \linear{R} \ldotp (q_P \in F_P \implies q_\Pi \in F_\Pi) \implies \\
    & f(R) \notin S \land ((q_P', \bot, (S \cup R(a)) \setminus D(a)), q_\Pi') \in \inactive(M_{P\Pi})
  \end{align*}
  which is indeed the case since $f(R) \notin S' \implies f(R) \notin S$ and $(S \cup R(a)) \setminus D(a) \subseteq (S' \cup R(a)) \setminus D(a)$.
\end{proof}

Note that this means that the antichain algorithm soundly returns potentially fewer counterexamples than the original one.

\subsection{Partition Optimization} \label{alg:partition}

The LTA construction for $\reduce_D(P)$ involves a nondeterministic choice of linear order at each state. Since $|\linear{\stmt}|$ has size $|\stmt|!$, each state in the automaton would have a large number of transitions. As an optimization, our algorithm  selects ordering relations out of $\partition{\stmt}$ (instead of $\linear{\stmt}$), defined as $\partition{\stmt} = \{ \stmt_1 \times \stmt_2 \mid \stmt_1 \uplus \stmt_2 = \stmt \}$ where $\uplus$ is disjoint union. This leads to a sound algorithm which is not complete with respect to sleep set reductions and trades the factorial complexity of computing $\linear{\stmt}$ for an exponential one.

%% file: cav19-extended/eval.tex
\section{Experiments}
\label{sec:eval}
\label{sec:imp}
\subsection{Implementation}

To evaluate our approach, we have implemented our algorithm in a tool called \tool written in Haskell. \tool accepts a program written in a simple imperative language as input, where the property is already encoded in the program in the form of {\em assume} statements,  and attempts to prove the program correct. The dependence relation for each input program is computed using a heuristic that ensures $\sim_D$-closedness. It is based on the fact that the shuffle product (i.e. parallel composition) of two $\sim_D$-closed languages is $\sim_D$-closed.

\tool employs two verification algorithms: (1) The total order algorithm presented in Algorithm \ref{alg:total}, and (2) the variation with the partition optimization discussed in Section \ref{alg:partition}. It also implements multiple counterexample generation algorithms: (1) \emph{Naive:} selects the first counterexample in the difference of the program and proof language. (2) \emph{Progress-Ensuring:} selects a set of counterexamples satisfying Theorem \ref{thm:finite-counterexamples}. (3) \emph{Bounded Progress-Ensuring:} selects a limited subset of the set computed by the progress-ensuring algorithm. We experimented with multiple strategies for the selection of the limited subset:
\begin{itemize}
  \item \emph{RR:} selects a lockstep trace of the program. In other words, it selects a single counterexample by choosing successive statements from each thread in a round robin fashion, e.g. the first statement is from thread 1, the second statement is from thread 2, and so on.
  \item \emph{Ln:} selects the $n$ leftmost counterexamples from the tree of counterexamples. When $n = 1$, this counterexample strategy effectively chooses sequential composition traces of the program.
  \item \emph{Mn:} selects the $n$ middlemost counterexamples from the tree of counterexamples. The middlemost counterexamples do not always contain lockstep counterexamples, so this strategy is distinct from the RR strategy.
\end{itemize}
Fixing a counterexample strategy does not mean committing to a particular reduction. For example, it is possible for assertions learned from a sequential composition trace to generalize to lockstep traces. In some of our benchmarks, we found that our algorithm converged using the L1 strategy, despite the fact that no full proof exists for the sequential composition reduction.

Our experimentation demonstrated that in all cases, the bounded progress ensuring algorithm (an instance of the partition algorithm) is the fastest of all options. Therefore, all our reports in this section are using this instance of the algorithm.

For the larger benchmarks, we use a simple sound optimization to reduce the proof size. We declare the basic blocks of code as atomic, so that internal assertions need not be generated for them as part of the proof. This optimization is incomplete with respect to sleep set reductions. \tool and all our benchmarks can be found at {\tt \url{github.com/weaver-verifier/weaver}}.



\subsection{Benchmarks}

We use a set of sequential benchmarks from \cite{SousaD16} and include  additional sequential benchmarks that involve more interesting reductions in their proofs.
We also have a set of parallel benchmarks, which are beyond the scope of previous hypersafety verification techniques. We use these benchmarks to demonstrate that our technique/tool can seamlessly handle concurrency. These involve proving concurrency specific hypersafety properties such as determinism and equivalence of parallel and sequential implementations of algorithms. Finally, since the proof checking algorithm is the core contribution of this paper,  we have a contrived set of instances to stress test our proof checking algorithm. These involve proving determinism of simple parallel-disjoint programs with various numbers of threads and statements per thread. These benchmarks have been designed to cause a combinatorial explosion for the proof checker and counterexample generation routines. The set of hypersafety properties used in all our experimental results are depicted in Figure \ref{fig:properties}.

\begin{figure}[htbp]
\begin{center}
\resizebox{\columnwidth}{!}{%
\begin{tabular}{|l|c|c|}
\hline
Property & Formula & Description \\
\hline
\textsc{CompSymm} & $\forall a, b, c \ldotp \operatorname{sign}(P(a, b)) = -\operatorname{sign}(P(b, a))$ & Comparator symmetry \\
\textsc{CompTrans} & $\forall a, b, c \ldotp P(a, b) > 0 \land P(b, c) > 0 \implies P(a, c) > 0$ & Comparator transitivity \\
\textsc{CompSubst} & $\forall a, b \ldotp P(a, b) = 0 \implies \operatorname{sign}(P(a, c)) = \operatorname{sign}(P(b, c))$ & Comparator  resp. equality \\
\textsc{Symm} & $\forall a, b \ldotp P(a, b) \implies P(b, a)$ & Symmetry \\
\textsc{Trans} & $\forall a, b, c \ldotp P(a, b) \land P(b, c) \implies P(a, c)$ & Transitivity \\
\textsc{Dist} & $\forall a, b, c \ldotp P(a + b, c) = P(a, c) + P(b, c)$ & Distributivity \\
\textsc{Sec} & $\forall h_1, h_2, l \ldotp P(h_1, l) = P(h_2, l)$ & Information flow security \\
\textsc{Det} & $\forall x, x'_1, x'_2 \ldotp (x, x'_1) \in P \land (x, x'_2) \in P \implies x'_1 = x'_2$ & Determinism \\
\textsc{Equiv} & $\forall a \ldotp P_1(a) = P_2(a)$ & Equivalence \\
\hline
\end{tabular}}
\caption{Hypersafety properties for our experiments. In the case of \textsc{Det}, $P$ is a relation between inputs and outputs rather than a function.}
\label{fig:properties}
\end{center}
\end{figure}

\paragraph{\bfseries Benchmarks from \cite{SousaD16}.}  Each example implements a comparator function that returns a negative number if the first argument is less than the second, a positive number if the first argument is greater than the second, and zero if they are equal. We check that each comparator satisfies \textsc{CompSubst}, \textsc{CompTrans}, and \textsc{CompSymm}.

\paragraph{\bfseries Our Sequential Benchmarks.} We verify the \textsc{Mult} example given in Section \ref{sec:example} satisfies \textsc{Dist} (i.e. $\textsc{Mult}(a + b, c) = \textsc{Mult}(a, c) + \textsc{Mult}(b, c)$). Since $\textsc{Mult}$ iterates on its first argument only, we also verify that the property holds when its arguments are flipped (i.e. $\textsc{Mult}(c, a + b) = \textsc{Mult}(c, a) + \textsc{Mult}(c, b)$). We also verify that an array equality procedure \textsc{ArrayEq} satisfies symmetry and transitivity, and include a simple information flow security example. Lastly, we include examples of loop unrolling: each \textsc{UnrollN} example involves a loop iterating a multiple of \textsc{N} times, and each \textsc{UnrollCondN} example unrolls a loop \textsc{N} times, with cleanup code for extra iterations.

\paragraph{\bfseries Our Parallel Benchmarks.} We use the following parallel benchmarks:
\begin{itemize}
\item The \textsc{Barrier} example checks that a simple loop-free barrier computation is deterministic.
\item The \textsc{Lamport} example verifies the correctness of a locking algorithm by checking whether executing a non-atomic operation in two threads within the confines of the lock is equivalent to executing the operation twice sequentially. The lock is implemented using Lamport's bakery algorithm.

%

\item The \textsc{ParallelSum} examples implement parallel summation over a queue. The \textsc{ParallelSum1} example has two threads atomically sum into a shared variable, while each thread in the \textsc{ParallelSum2} example sums into a local variable. We verify that they are deterministic. We also verify that \textsc{ParallelSum1} is equivalent to a single-threaded version of the same program.


\item The \textsc{SimpleInc} example verifies that atomically incrementing an integer in two different threads is equivalent to non-atomically incrementing it twice in a single thread. The \textsc{Spaghetti} example verifies the determinism of a program that performs an arbitrary computation before setting its output to a fixed value.

\item The \textsc{ExpNxM} benchmarks verify that a program computing an exponential term is deterministic. Each program is replicated over \textsc{N} threads, and each thread contains \textsc{M} statements.
\end{itemize}

\subsection{Evaluation}

Detailed results of our experiments are included in Appendix \ref{app:results}. Table \ref{tab:1} includes a summary in the form of averages, and here, we discuss our top findings.

\begin{table}[t]
\begin{center}
\begin{tabular}{|l|c|c|c|c|c|c|}\hline
\multirow{3}{*}{\shortstack{Benchmark Group}} &
\multirow{3}{*}{\shortstack{Group\\Count}} &
\multirow{3}{*}{\shortstack{Proof\\Size}} &
\multirow{3}{*}{\shortstack{Number of\\Refinement\\Rounds}} &
\multirow{3}{*}{\shortstack{Proof\\Construction\\Time}} &
\multirow{3}{*}{\shortstack{Proof\\Checking\\Time}} &
\multirow{3}{*}{\shortstack{Total\\Time}} \\
& & & & & & \\
& & & & & & \\ \hline \hline
\multirow{3}{*}{\shortstack{
Looping programs of \cite{SousaD16}\\ 2-safety properties}} & \multirow{3}{*}{5} & \multirow{3}{*}{63} & \multirow{3}{*}{12} & \multirow{3}{*}{46.69s} & \multirow{3}{*}{0.1s} & \multirow{3}{*}{47.03s} \\
& & & & & & \\
& & & & & & \\ \hline
\multirow{3}{*}{\shortstack{
Looping programs of \cite{SousaD16}\\ 3-safety properties}} & \multirow{3}{*}{8} & \multirow{3}{*}{155} & \multirow{3}{*}{22} & \multirow{3}{*}{475.78s} & \multirow{3}{*}{11.79s} & \multirow{3}{*}{448.36s} \\
& & & & & & \\
& & & & & & \\ \hline

Loop-free programs of \cite{SousaD16} & 27 & 5 & 2 & 0.13s & 0.0004s & 0.15s \\ \hline
Our sequential benchmarks & 13 & 30 & 9 & 14.27s & 2.5s & 17.94s \\ \hline
Our parallel benchmarks& 7 & 31 & 8 & 17.95 & 0.56s & 18.63s \\ \hline
\end{tabular} 
\caption{Experimental results averages for benchmark groups. \label{tab:1}}
\end{center}
\vspace{-30pt}
\end{table}%


{\bf Proof construction time} refers to the time spent to construct $\lang{\Pi}$ from a given set of assertions $\Pi$ and excludes the time to produce proofs for the counterexamples in a given round.  {\bf Proof checking time} is the time spent to check if the current proof candidate  is strong enough for a reduction of the program. In the fastest instances (total time around $0.01$ seconds), roughly equal time is spent in proof checking and proof construction. In the slowest instances, the total time is almost entirely spent in proof construction. In contrast, in our stress tests (designed to stress the proof checking algorithm) the majority of the time is spent in proof checking. The time spent in proving counterexamples correct is negligible in all instances.  {\bf Proof sizes} vary from 4 assertions to 298 for the most complicated instance. Verification times are {\em correlated} with the final proof size; larger proofs tend to cause longer verification times.

{\bf Numbers of refinement rounds} vary from 2 for the simplest to 33 for the most complicated instance. A small number of refinement rounds (e.g. 2) implies a fast verification time. But, for the  higher number of rounds, a strong positive correlation between the number of rounds and verification time does not exist.

The bounded progress ensuring algorithm achieves its  best time, in most cases, by taking 1 counterexample at each refinement round. There are two exceptions which respectively require 5 counterexamples from the left (the parallel sum) and middle (transitivity of integer array comparator) per refinement round for the best overall verification time. The former is negligible because the 1 counterexample variation produces a runtime within 12\% of this optimal answer. The latter is significant, since other instances take substantially more time to verify.

In the vast majority of the cases, choosing a counterexample from the left or the middle makes a negligible difference in the verification time. The exceptions are the hardest sequential instance and the hardest parallel one. In these two cases, a left choice leads to a timeout while the middle one succeeds.

To gauge the effect of the counterexample selection strategy on the overall performance of our approach, we used all three counterexample selection modes on all our benchmarks: RR, L$n$, and M$n$ , for $n \in \{1, 5, 10\}$. In most cases, the fastest configuration is 1-2x faster than the next fastest. For the more complicated benchmarks, the difference is much larger; for example, \tool converges 13x faster using the RR compared to M$1$ counterexample mode for the substitutivity of the integer array comparator.

Our conclusion is that in the cases that the proof is complicated and can grow large, the choice of counterexamples can have a big impact on verification time. One can simply run a few  instantiations of the algorithm with different counterexample choices in parallel and wait for the best one to return.

For our {\bfseries parallel programs} benchmarks  (other than our stress tests), the tool spends the majority of its time in proof construction. Therefore, we designed specific (unusual) parallel programs to stress test the proof checker. {\bf Stress test} benchmarks are trivial tests of determinism of disjoint parallel programs, which can be proven correct easily by using the atomic block optimization. However, we force the tool to do the unnecessary hard work. These instances simulate the worst case theoretical complexity where the proof checking time and number of counterexamples grow exponentially with the number of threads and the sizes of the threads. In the largest instance, more than 99\% of the total verification time is spent in proof checking. Averages are not very informative for these instances, and therefore are not included in Table \ref{tab:1}.

Finally, \tool is only slow for verifying 3-safety properties of large looping benchmarks from \cite{SousaD16}. Note that unlike the approach in \cite{SousaD16}, which starts from a default lockstep reduction (that is incidentally sufficient to prove these instances), we do not assume any reduction and consider them all. The extra time is therefore expected when the product programs become quite large.

{\bfseries Antichains.} To measure the effectiveness of our antichain optimization, we compare the times for proof checking of the final (correct) proof for all our benchmarks with and without the antichain optimization. We discard the cases where the proof checking time is under 0.01s without optimization as the times are too small to be statistically relevant. The antichain-based algorithm outperforms the unoptimized one in all these instances. Of the remaining (27 benchamrks), the minimum speedup is a factor of 2, and the maximum is a factor of 44 (not considering the one instance that times out without the optimization). The average speedup is 12, and larger speedups (on average) are  observed for the parallel benchmarks compared to the sequential ones.

%

%
%
%
%

%% file: cav19-extended/related.tex

\section{Related Work}

The notion of a $k$-safety hyperproperty was introduced in \cite{ClarksonS08} without consideration for automatic program verification.
The approach of reducing $k$-safety to 1-safety by self-composition is introduced in \cite{BartheDR11}. While theoretically complete, self-composition is not practical as discussed in Section \ref{sec:intro}.
Product programs generalize the self-composition approach and have been used in verifying translation validation \cite{PnueliSS98}, non-interference \cite{GoguenM82a,SabelfeldM03}, and program optimization \cite{SousaDVDG14}. A product of two programs $P_1$ and $P_2$ is semantically equivalent to $P_1 \cdot P_2$ (sequential composition), but is made easier to verify by allowing parts of each program to be interleaved. The product programs proposed in \cite{BartheCK11} allow lockstep interleaving exclusively, but only when the control structures of $P_1$ and $P_2$ match. This restriction is lifted in \cite{BartheCK13} to allow some non-lockstep interleavings. However, the given construction rules are non-deterministic, and the choice of product program is left to the user or a heuristic.

Relational program logics \cite{Benton04,Yang07} extend traditional program logics to allow reasoning about relational program properties, however automation is usually not addressed. Automatic construction of product programs is discussed in \cite{EilersMH18} with the goal of supporting procedure specifications and modular reasoning, but is also restricted to lockstep interleavings. Our approach does not support procedure calls but is fully automated and permits non-lockstep interleavings.

The key feature of our approached is the automation of the discovery of an appropriate program reduction and a proof combined. In this case, the only other method that compares is the one based on Cartesian Hoare Logic (CHL) proposed in \cite{SousaD16} along with an algorithm for automatic verification based on CHL. Their proposed algorithm implicitly constructs a product program, using a heuristic that favours lockstep executions as much as possible, and then prioritizes certain rules of the logic over the rest. The heuristic nature of the search for the proof means that no characterization of the search space can be given, and no guarantees about whether an appropriate product program will be found. In contrast, we have a formal characterization of the set of explored product programs in this paper. Moreover, CHL was not designed to deal with concurrency.

Lipton \cite{Lipton75} first proposed reduction as a way to simplify reasoning about concurrent programs. His ideas have been employed in a semi-automatic setting in \cite{ElmasQT09}. Partial-order reduction (POR) is a class of techniques that reduces the state space of search by removing redundant paths. POR techniques are concerned with finding a single (preferably minimal) reduction of the input program. In contrast, we use the same underlying ideas to explore many program reductions simultaenously. The class of reductions described in Section \ref{sec:sleep} is based on the sleep set technique of Godefroid \cite{Godefroid96}. Other techniques exist \cite{Godefroid96,AbdullaAJS17} that are used in conjunction with sleep sets to achieve minimality in a normal POR setting. In our setting, reductions generated by sleep sets are already optimal (Theorem \ref{thm:sleep-optimal}). However, employing these additional POR techniques may propose ways of optimizing our proof checking algorithm by producing a smaller reduction LTA.

%% file: cav19-extended/appendix.tex

\newpage

\section{Results}\label{app:results}
Benchmarks were run on a Proliant DL980 G7 with eight eight-core Intel X6550 processors (64 cores, 64 threads) and 256G of RAM, running 64-bit Ubuntu.

\begin{table}[ht]
\resizebox{\textwidth}{!}{\begin{tabular}{|l|c|c|c|c|c|c|c|}
\hline
Benchmark & Property & Algorithm & \shortstack{Proof \\ Size} & \shortstack{Refinement \\ Rounds} & \shortstack{Proof \\ Construction \\ Time} & \shortstack{Proof \\ Checking \\ Time} & \shortstack{Total \\ Time} \\
\hline

\multicolumn{8}{c}{Our Sequential Benchmarks} \\
\hline
\textsc{ArrayEq}     & \textsc{Symm}           & BPE(RR) & 17 & 10 & 1.17s   & 0.16s  & 1.38s \\
\textsc{ArrayEq}     & \textsc{Trans}          & BPE(M5) & 97 & 15 & 156.17s & 30.71s & 201.2s \\
\textsc{Mult}        & \textsc{Dist}           & BPE(M1) & 21 & 11 & 1.34s   & 1.32s  & 2.94s \\
\textsc{Mult}        & \textsc{Dist} (Flipped) & BPE(L1) & 27 & 12 & 2.54s   & 1.38s  & 4s \\
\textsc{Security}    & \textsc{Sec}            & BPE(RR) & 12 & 6  & 0.41s   & 0.052s & 0.49s \\
\textsc{Unroll2}     & \textsc{Equiv}          & BPE(RR) & 21 & 9  & 1.42s   & 0.017s & 1.48s \\
\textsc{Unroll3}     & \textsc{Equiv}          & BPE(RR) & 25 & 9  & 1.99s   & 0.02s  & 2.06s \\
\textsc{Unroll4}     & \textsc{Equiv}          & BPE(L1) & 31 & 11 & 3.31s   & 0.05s  & 3.43s \\
\textsc{Unroll5}     & \textsc{Equiv}          & BPE(L1) & 36 & 12 & 4.87s   & 0.06s  & 5.01s \\
\textsc{UnrollCond2} & \textsc{Equiv}          & BPE(RR) & 22 & 7  & 0.85s   & 0.01s  & 0.88s \\
\textsc{UnrollCond3} & \textsc{Equiv}          & BPE(RR) & 30 & 9  & 2.28s   & 0.02s  & 2.35s \\
\textsc{UnrollCond4} & \textsc{Equiv}          & BPE(L1) & 38 & 9  & 4.7s    & 0.02s  & 4.8s \\
\textsc{UnrollCond5} & \textsc{Equiv}          & BPE(L1) & 45 & 10 & 7.06s   & 0.02s  & 7.18s \\
\hline
\multicolumn{8}{c}{Our Parallel Benchmarks} \\
\hline
\textsc{Barrier}      & \textsc{Det}   & BPE(RR) & 21 & 6  & 4.05s   & 1.98s   & 6.1s \\
\textsc{Lamport}      & \textsc{Equiv} & BPE(RR) & 28 & 8  & 3.24s   & 0.44s   & 3.77s \\
\textsc{ParallelSum1} & \textsc{Equiv} & BPE(M5) & 32 & 7  & 6.01s   & 0.19s   & 6.45s \\
\textsc{ParallelSum1} & \textsc{Det}   & BPE(RR) & 37 & 13 & 10.38s  & 0.22s   & 10.72s \\
\textsc{ParallelSum2} & \textsc{Det}   & BPE(RR) & 86 & 15 & 101.89s & 1.07s   & 103.23s \\
\textsc{SimpleInc}    & \textsc{Equiv} & BPE(M1) & 7  & 2  & 0.05s   & 0.0008s & 0.06s \\
\textsc{Spaghetti}    & \textsc{Det}   & BPE(RR) & 4  & 2  & 0.03s   & 0.02s   & 0.05s \\
\hline
\multicolumn{8}{c}{Stress Tests} \\
\hline
\textsc{Exp1x3} & \textsc{Det} & BPE(L1) & 10 & 5  & 0.13s & 0.005s   & 0.15s \\
\textsc{Exp2x3} & \textsc{Det} & BPE(L1) & 17 & 8  & 0.88s & 3.03s    & 3.95s \\
\textsc{Exp2x4} & \textsc{Det} & BPE(RR) & 18 & 8  & 1.11s & 4.23s    & 5.39s \\
\textsc{Exp2x6} & \textsc{Det} & BPE(RR) & 19 & 8  & 1.32s & 8.26s    & 9.65s \\
\textsc{Exp2x9} & \textsc{Det} & BPE(RR) & 19 & 8  & 1.34s & 19.19s   & 20.61s \\
\textsc{Exp3x3} & \textsc{Det} & BPE(L5) & 26 & 11 & 3.77s & 1672.05s & 1676.61s \\
\hline
\end{tabular}}
\caption{The detailed results of the winning algorithm performed on our benchmarks. TO = total-order, P = partition, N = naive, BPE(?n) = bounded-progress-ensuring (with $n$ counterexamples from left (L), middle (M), or round-robin (RR)).}
\end{table}

\begin{table}[hbt!]
\resizebox{\textwidth}{!}{\begin{tabular}{|l|c|c|c|c|c|c|c|}
\hline
Benchmark & Property & Algorithm & \shortstack{Proof \\ Size} & \shortstack{Refinement \\ Rounds} & \shortstack{Proof \\ Construction \\ Time} & \shortstack{Proof \\ Checking \\ Time} & \shortstack{Total \\ Time} \\
\hline
\hline
\textsc{ArrayInt}       & \textsc{CompSubst} & BPE(RR) & 146 & 20 & 517.33s  & 5.22    & 523.82 \\
\textsc{ArrayInt}       & \textsc{CompSymm}  & BPE(RR) & 29  & 8  & 6.21s    & 0.02s   & 6.36s \\
\textsc{ArrayInt}       & \textsc{CompTrans} & BPE(RR) & 148 & 21 & 390.15s  & 2.27s   & 393.66s \\
\textsc{Chromosome}     & \textsc{CompSubst} & BPE(L1) & 136 & 18 & 211.34s  & 11.39s  & 223.58s \\
\textsc{Chromosome}     & \textsc{CompSymm}  & BPE(RR) & 99  & 16 & 61.68s   & 0.15s   & 62.2s \\
\textsc{Chromosome}     & \textsc{CompTrans} & BPE(L1) & 179 & 23 & 274s     & 16.56s  & 291.22s \\
\textsc{Collitem}       & \textsc{CompSubst} & BPE(M1) & 5   & 2  & 0.5s     & 0.0004s & 0.55s \\
\textsc{Collitem}       & \textsc{CompSymm}  & BPE(M1) & 4   & 2  & 0.09s    & 0.0003s & 0.12s \\
\textsc{Collitem}       & \textsc{CompTrans} & BPE(M1) & 5   & 2  & 0.21s    & 0.0003s & 0.24s \\
\textsc{Container}      & \textsc{CompSubst} & BPE(M1) & 5   & 2  & 0.39s    & 0.0004s & 0.43s \\
\textsc{Container}      & \textsc{CompSymm}  & BPE(M1) & 4   & 2  & 0.09s    & 0.001s  & 0.12s \\
\textsc{Container}      & \textsc{CompTrans} & BPE(M1) & 5   & 2  & 0.22s    & 0.0005s & 0.25s \\
\textsc{ExpTerm}        & \textsc{CompSubst} & BPE(L1) & 5   & 2  & 0.16s    & 0.0005s & 0.18s \\
\textsc{ExpTerm}        & \textsc{CompSymm}  & BPE(M1) & 4   & 2  & 0.07s    & 0.0003s & 0.08s \\
\textsc{ExpTerm}        & \textsc{CompTrans} & BPE(M1) & 4   & 2  & 0.08s    & 0.0003s & 0.09s \\
\textsc{FileItem}       & \textsc{CompSubst} & BPE(RR) & 5   & 2  & 0.14s    & 0.0004s & 0.16s \\
\textsc{FileItem}       & \textsc{CompSymm}  & BPE(RR) & 4   & 2  & 0.07s    & 0.0003s & 0.08s \\
\textsc{FileItem}       & \textsc{CompTrans} & BPE(L1) & 5   & 2  & 0.12s    & 0.0004s & 0.13s \\
\textsc{Match}          & \textsc{CompSubst} & BPE(RR) & 5   & 2  & 0.11s    & 0.0004s & 0.12s \\
\textsc{Match}          & \textsc{CompSymm}  & BPE(M1) & 4   & 2  & 0.06s    & 0.0005s & 0.07s \\
\textsc{Match}          & \textsc{CompTrans} & BPE(RR) & 5   & 2  & 0.07s    & 0.0004s & 0.08s \\
\textsc{NameComparator} & \textsc{CompSubst} & BPE(L1) & 75  & 14 & 70.41s   & 0.37s   & 71.03s \\
\textsc{NameComparator} & \textsc{CompSymm}  & BPE(L1) & 49  & 9  & 15.5s    & 0.026s  & 15.66s \\
\textsc{NameComparator} & \textsc{CompTrans} & BPE(L1) & 114 & 18 & 114.88s  & 3.56s   & 118.84s \\
\textsc{Node}           & \textsc{CompSubst} & BPE(L1) & 5   & 2  & 0.27s    & 0.0004s & 0.3s \\
\textsc{Node}           & \textsc{CompSymm}  & BPE(L1) & 4   & 2  & 0.06s    & 0.0002s & 0.08s \\
\textsc{Node}           & \textsc{CompTrans} & BPE(RR) & 5   & 2  & 0.1s     & 0.0004s & 0.11s \\
\textsc{NzbFile}        & \textsc{CompSubst} & BPE(RR) & 298 & 33 & 2026.73s & 28.51s  & 2056.9s \\
\textsc{NzbFile}        & \textsc{CompSymm}  & BPE(RR) & 111 & 19 & 144.96s  & 0.27s   & 145.72s \\
\textsc{NzbFile}        & \textsc{CompTrans} & BPE(RR) & 219 & 29 & 751.71s  & 43.97s  & 796.54s \\
\textsc{SimplStr}       & \textsc{CompSubst} & BPE(RR) & 5   & 2  & 0.12s    & 0.0003s & 0.13s \\
\textsc{SimplStr}       & \textsc{CompSymm}  & BPE(L1) & 4   & 2  & 0.03s    & 0.0002s & 0.04s \\
\textsc{SimplStr}       & \textsc{CompTrans} & BPE(RR) & 5   & 2  & 0.07s    & 0.0003s & 0.07s \\
\textsc{Sponsored}      & \textsc{CompSubst} & BPE(RR) & 5   & 2  & 0.07s    & 0.0003s & 0.08s \\
\textsc{Sponsored}      & \textsc{CompSymm}  & BPE(M1) & 4   & 2  & 0.03s    & 0.0002s & 0.03s \\
\textsc{Sponsored}      & \textsc{CompTrans} & BPE(M1) & 5   & 2  & 0.03s    & 0.0003s & 0.04s \\
\textsc{Time}           & \textsc{CompSubst} & BPE(M1) & 5   & 2  & 0.17s    & 0.0005s & 0.19s \\
\textsc{Time}           & \textsc{CompSymm}  & BPE(M1) & 4   & 2  & 0.03s    & 0.0005s & 0.04s \\
\textsc{Time}           & \textsc{CompTrans} & BPE(M1) & 5   & 2  & 0.07s    & 0.0004s & 0.08s \\
\textsc{Word}           & \textsc{CompSubst} & BPE(RR) & 103 & 18 & 179.29s  & 2.59s   & 182.28s \\
\textsc{Word}           & \textsc{CompSymm}  & BPE(RR) & 26  & 8  & 5.14s    & 0.017s  & 5.22s \\
\textsc{Word}           & \textsc{CompTrans} & BPE(RR) & 129 & 22 & 221.94s  & 3.48s   & 225.78s \\
\hline
\end{tabular}}
\caption{The detailed results of the winning algorithm performed on benchmarks from \cite{SousaD16}. TO = total-order, P = partition, N = naive, BPE(?n) = bounded-progress-ensuring (with $n$ counterexamples from left (L), middle (M), or round-robin (RR))}
\end{table}

\begin{table}[ht]\label{tab:3}
  \resizebox{\textwidth}{!}{\begin{tabular}{|l|c|c|c|c|c|}
\hline
Benchmark & Property & Threads & Proof Size & Optimized Time & Unoptimized Time \\
\hline
\multicolumn{6}{c}{Our Sequential Benchmarks} \\
\hline
\textsc{ArrayEq}     & \textsc{Symm}           & 2 & 17 & 0.03s  & 0.13s \\
\textsc{ArrayEq}     & \textsc{Trans}          & 3 & 97 & 5.99s  & 75.53s \\
\textsc{Mult}        & \textsc{Dist}           & 3 & 21 & 0.42s  & 5.21s \\
\textsc{Mult}        & \textsc{Dist} (Flipped) & 3 & 27 & 0.33s  & 3.11s \\
\textsc{Security}    & \textsc{Sec}            & 2 & 12 & 0.01s  & 0.06s \\
\textsc{Unroll2}     & \textsc{Equiv}          & 2 & 21 & 0.003s & 0.007s \\
\textsc{Unroll3}     & \textsc{Equiv}          & 2 & 25 & 0.005s & 0.01s \\
\textsc{Unroll4}     & \textsc{Equiv}          & 2 & 31 & 0.01s  & 0.05s \\
\textsc{Unroll5}     & \textsc{Equiv}          & 2 & 36 & 0.01s  & 0.05s \\
\textsc{UnrollCond2} & \textsc{Equiv}          & 2 & 22 & 0.002s & 0.003s \\
\textsc{UnrollCond3} & \textsc{Equiv}          & 2 & 30 & 0.003s & 0.005s \\
\textsc{UnrollCond4} & \textsc{Equiv}          & 2 & 38 & 0.007s & 0.02s \\
\textsc{UnrollCond5} & \textsc{Equiv}          & 2 & 45 & 0.006s & 0.01s \\
\hline
\multicolumn{6}{c}{Our Parallel Benchmarks} \\
\hline
\textsc{Barrier}      & \textsc{Det}   & 4 & 21  & 0.47s & 11.46s \\
\textsc{Lamport}      & \textsc{Equiv} & 3 & 28  & 0.09s & 0.36s \\
\textsc{ParallelSum1} & \textsc{Equiv} & 3 & 32  & 0.09s & 1.10s \\
\textsc{ParallelSum1} & \textsc{Det}   & 4 & 37  & 0.12s & 1.17s \\
\textsc{ParallelSum2} & \textsc{Det}   & 4 & 86  & 0.1s & 1.21s \\
\textsc{SimpleInc}    & \textsc{Equiv} & 3 & 7   & 0.0004s & 0.001s \\
\textsc{Spaghetti}    & \textsc{Det}   & 4 & 4   & 0.004s & 2.28s \\
\hline
\multicolumn{6}{c}{Stress Tests} \\
\hline
\textsc{Exp1x3} & \textsc{Det} & 2 & 10 & 0.002s  & 0.005s \\
\textsc{Exp2x3} & \textsc{Det} & 4 & 17 & 1.15s   & 27.2s \\
\textsc{Exp2x4} & \textsc{Det} & 4 & 18 & 1.33s   & 36.009s \\
\textsc{Exp2x6} & \textsc{Det} & 4 & 19 & 2.34s   & 80.98s \\
\textsc{Exp2x9} & \textsc{Det} & 4 & 19 & 5.045s  & 222.45s \\
\textsc{Exp3x3} & \textsc{Det} & 6 & 26 & 910.03s & T/O \\
\hline
\multicolumn{6}{c}{Benchmarks from \cite{SousaD16}} \\
\hline
\textsc{ArrayInt}       & \textsc{CompSubst} & 3 & 146 & 0.54s     & 2.47s \\
\textsc{ArrayInt}       & \textsc{CompSymm}  & 2 & 29  & 0.004s    & 0.005s \\
\textsc{ArrayInt}       & \textsc{CompTrans} & 3 & 148 & 0.13s     & 0.71s \\
\textsc{Chromosome}     & \textsc{CompSubst} & 3 & 136 & 2.81s     & 32.63s \\
\textsc{Chromosome}     & \textsc{CompSymm}  & 2 & 99  & 0.02s     & 0.02s \\
\textsc{Chromosome}     & \textsc{CompTrans} & 3 & 179 & 3.25s     & 24.55s \\
\textsc{NameComparator} & \textsc{CompSubst} & 3 & 75  & 0.02s     & 0.31s \\
\textsc{NameComparator} & \textsc{CompSymm}  & 2 & 49  & 0.004s    & 0.009s \\
\textsc{NameComparator} & \textsc{CompTrans} & 3 & 114 & 0.22s     & 1.79s \\
\textsc{NzbFile}        & \textsc{CompSubst} & 3 & 298 & 6.25s     & 34.38s \\
\textsc{NzbFile}        & \textsc{CompSymm}  & 2 & 111 & 0.041s    & 0.082s \\
\textsc{NzbFile}        & \textsc{CompTrans} & 3 & 219 & 9.39s     & 50.74s \\
\textsc{Word}           & \textsc{CompSubst} & 3 & 103 & 0.42s     & 2.11s \\
\textsc{Word}           & \textsc{CompSymm}  & 2 & 26  & 0.005s    & 0.007s \\
\textsc{Word}           & \textsc{CompTrans} & 3 & 129 & 0.39s     & 2.038s \\
\hline
\end{tabular}}
\caption{The detailed results of the last round of proof checking for all benchmarks, with and without the antichain optimization. The value of the Threads column is the number of threads in the encoded program. Benchmarks whose proofs take less than 0.01 seconds to check are omitted.}
\end{table}

%% file: cav19-extended.bbl
\begin{thebibliography}{10}
\providecommand{\url}[1]{\texttt{#1}}
\providecommand{\urlprefix}{URL }
\providecommand{\doi}[1]{https://doi.org/#1}

\bibitem{AbdullaAJS17}
Abdulla, P.A., Aronis, S., Jonsson, B., Sagonas, K.: Source sets: a foundation
  for optimal dynamic partial order reduction. Journal of the ACM (JACM)
  \textbf{64}(4), ~25 (2017)

\bibitem{BaaderT01}
Baader, F., Tobies, S.: The inverse method implements the automata approach for
  modal satisfiability. In: International Joint Conference on Automated
  Reasoning. pp. 92--106. Springer (2001)

\bibitem{BartheCK11}
Barthe, G., Crespo, J.M., Kunz, C.: Relational verification using product
  programs. In: International Symposium on Formal Methods. pp. 200--214.
  Springer (2011)

\bibitem{BartheCK13}
Barthe, G., Crespo, J.M., Kunz, C.: Beyond 2-safety: Asymmetric product
  programs for relational program verification. In: International Symposium on
  Logical Foundations of Computer Science. pp. 29--43. Springer (2013)

\bibitem{BartheDR11}
Barthe, G., D'argenio, P.R., Rezk, T.: Secure information flow by
  self-composition. Mathematical Structures in Computer Science
  \textbf{21}(6),  1207--1252 (2011)

\bibitem{Benton04}
Benton, N.: Simple relational correctness proofs for static analyses and
  program transformations. In: ACM SIGPLAN Notices. vol.~39, pp. 14--25. ACM
  (2004)

\bibitem{ClarksonS08}
Clarkson, M.R., Schneider, F.B.: Hyperproperties. In: 21st IEEE Computer
  Security Foundations Symposium. pp. 51--65. IEEE (2008)

\bibitem{WulfDHR06}
De~Wulf, M., Doyen, L., Henzinger, T.A., Raskin, J.F.: Antichains: A new
  algorithm for checking universality of finite automata. In: International
  Conference on Computer Aided Verification. pp. 17--30. Springer (2006)

\bibitem{DiekertM97}
Diekert, V., M{\'e}tivier, Y.: Partial commutation and traces. In: Handbook of
  formal languages, pp. 457--533. Springer (1997)

\bibitem{EilersMH18}
Eilers, M., M{\"u}ller, P., Hitz, S.: Modular product programs. In: European
  Symposium on Programming. pp. 502--529. Springer (2018)

\bibitem{ElmasQT09}
Elmas, T., Qadeer, S., Tasiran, S.: A calculus of atomic actions. In: ACM
  SIGPLAN Notices. vol.~44, pp. 2--15. ACM (2009)

\bibitem{FarzanKP13}
Farzan, A., Kincaid, Z., Podelski, A.: Inductive data flow graphs. In: ACM
  SIGPLAN Notices. vol.~48, pp. 129--142. ACM (2013)

\bibitem{FarzanKP15}
Farzan, A., Kincaid, Z., Podelski, A.: Proof spaces for unbounded parallelism.
  In: ACM SIGPLAN Notices. vol.~50, pp. 407--420. ACM (2015)

\bibitem{Godefroid96}
Godefroid, P., Van~Leeuwen, J., Hartmanis, J., Goos, G., Wolper, P.:
  Partial-order methods for the verification of concurrent systems: an approach
  to the state-explosion problem, vol.~1032. Springer Heidelberg (1996)

\bibitem{GoguenM82a}
Goguen, J.A., Meseguer, J.: Security policies and security models. In: Security
  and Privacy, 1982 IEEE Symposium on. pp. 11--11. IEEE (1982)

\bibitem{HeizmannHP09}
Heizmann, M., Hoenicke, J., Podelski, A.: Refinement of trace abstraction. In:
  International Static Analysis Symposium. pp. 69--85. Springer (2009)

\bibitem{Hopcroft:2006}
Hopcroft, J.E., Motwani, R., Ullman, J.D.: Introduction to Automata Theory,
  Languages, and Computation (3rd Edition). Addison-Wesley Longman Publishing
  Co., Inc., Boston, MA, USA (2006)

\bibitem{Lipton75}
Lipton, R.J.: Reduction: A method of proving properties of parallel programs.
  Communications of the ACM  \textbf{18}(12),  717--721 (1975)

\bibitem{PnueliSS98}
Pnueli, A., Siegel, M., Singerman, E.: Translation validation. In:
  International Conference on Tools and Algorithms for the Construction and
  Analysis of Systems. pp. 151--166. Springer (1998)

\bibitem{PopeeaRW14}
Popeea, C., Rybalchenko, A., Wilhelm, A.: Reduction for compositional
  verification of multi-threaded programs. In: Formal Methods in Computer-Aided
  Design (FMCAD), 2014. pp. 187--194. IEEE (2014)

\bibitem{Pottier_lazyleast}
Pottier, F.: Lazy least fixed points in ml

\bibitem{SabelfeldM03}
Sabelfeld, A., Myers, A.C.: Language-based information-flow security. IEEE
  Journal on selected areas in communications  \textbf{21}(1),  5--19 (2003)

\bibitem{SousaD16}
Sousa, M., Dillig, I.: Cartesian hoare logic for verifying k-safety properties.
  In: ACM SIGPLAN Notices. vol.~51, pp. 57--69. ACM (2016)

\bibitem{SousaDVDG14}
Sousa, M., Dillig, I., Vytiniotis, D., Dillig, T., Gkantsidis, C.:
  Consolidation of queries with user-defined functions. In: ACM SIGPLAN
  Notices. vol.~49, pp. 554--564. ACM (2014)

\bibitem{TerauchiA05}
Terauchi, T., Aiken, A.: Secure information flow as a safety problem. In:
  International Static Analysis Symposium. pp. 352--367. Springer (2005)

\bibitem{VardiW94}
Vardi, M.Y., Wolper, P.: Reasoning about infinite computations. Information \&
  Computation  \textbf{115}(1),  1--37 (1994)

\bibitem{Yang07}
Yang, H.: Relational separation logic. Theoretical Computer Science
  \textbf{375}(1-3),  308--334 (2007)

\end{thebibliography}
